\documentclass[a4]{llncs}

\usepackage[usenames]{color}
\usepackage{thm-restate}
\usepackage{amsmath}
\usepackage{amstext}
\usepackage{amsfonts}
\usepackage{mathtools}
\usepackage{paralist}
\usepackage{array}
\usepackage{listings}
\usepackage{courier}
\usepackage{algorithm2e}
\usepackage{tikz}
\usepackage{pgfplots} 
\usepackage{times}
\usepackage[h]{esvect}
\usepackage[strict]{changepage}

\usetikzlibrary{arrows}

\newcommand{\hide}[1]{}
\newcommand{\td}[1]{\noindent\textcolor{blue}{\ifmmode \text{[#1]}\else [#1] \fi}}

\lstdefinelanguage{scala}{
  morekeywords={abstract,case,catch,class,def,%
    do,else,extends,false,final,finally,%
    for,if,implicit,import,match,mixin,%
    new,null,object,override,package,%
    private,protected,requires,return,sealed,%
    super,this,throw,trait,true,try,%
    type,val,var,while,with,yield},
  otherkeywords={=>,<-,<\%,<:,>:,\#,@},
  sensitive=true,
  morecomment=[l]{//},
  morecomment=[n]{/*}{*/},
  morestring=[b]",
  morestring=[b]',
  morestring=[b]"""
}

\newcommand{\kw}[1]{\textbf{#1}}
\newcommand{\li}[0]{p}
\newcommand{\foldl}[0]{{\small\textsf{\kw{foldl}}}}

\newcommand{\zero}[0]{z}
\newcommand{\apply}[0]{\mathit{apply}}
\newcommand{\seq}[0]{\mathit{seq}}
\newcommand{\comb}[0]{\mathit{comb}}
\newcommand{\listof}[1]{[#1]}

\newcommand{\rdd}[0]{\mathit{rdd}}
\newcommand{\RDD}[0]{{\small\textsf{RDD}}}
\newcommand{\PairRDD}[0]{{\small\textsf{PairRDD}}}
\newcommand{\GraphRDD}[0]{{\small\textsf{GraphRDD}}}
\newcommand{\prdd}[0]{\mathit{prdd}}
\newcommand{\rddof}[1]{\rdd(#1)}

\newcommand{\graphrdd}[0]{\mathit{graphRdd}}
\newcommand{\graphrddof}[1]{\graphrdd(#1)}
\newcommand{\vertexid}[0]{\mathit{VertexID}}
\newcommand{\vertexlist}[0]{V}
\newcommand{\edgelist}[0]{E}
\newcommand{\numTri}[1]{\triangle_{#1}}
\newcommand{\agg}[0]{{\small\textsf{\kw{aggregate}}}}
\newcommand{\red}[0]{{\small\textsf{\kw{reduce}}}}
\newcommand{\redl}[0]{{\small\textsf{\kw{reducel}}}}

\newcommand{\redListWKey}[0]{{\small\textsf{\kw{reduceListWithKey}}}}
\newcommand{\treered}[0]{{\small\textsf{\kw{treeReduce}}}}
\newcommand{\nothing}[0]{{\small\textsf{\kw{Nothing}}}}
\newcommand{\maybe}[0]{{\small\textsf{\kw{Maybe}}}}
\newcommand{\just}[0]{{\small\textsf{\kw{Just}}}}
\newcommand{\fromJust}[0]{{\small\textsf{\kw{fromJust}}}}

\newcommand{\treeagg}[0]{{\small\textsf{\kw{treeAggregate}}}}
\newcommand{\aggWKey}[0]{{\small\textsf{\kw{aggregateWithKey}}}}
\newcommand{\redWKey}[0]{{\small\textsf{\kw{reduceWithKey}}}}
\newcommand{\aggMsgs}[0]{{\small\textsf{\kw{aggregateMessages}}}}

\newcommand{\aggBKey}[0]{{\small\textsf{\kw{aggregateByKey}}}}

\newcommand{\redBKey}[0]{{\small\textsf{\kw{reduceByKey}}}}

\newcommand{\perm}[0]{\mathit{perm}}
\newcommand{\send}[0]{\mathit{send}}
\newcommand{\lookup}[0]{{\small\textsf{\kw{lookup}}}}
\newcommand{\lookUp}[0]{{\small\textsf{\kw{lookUp}}}}
\newcommand{\map}[0]{{\small\textsf{\kw{map}}}}
\newcommand{\xnot}[0]{{\small\textsf{\kw{not}}}}
\newcommand{\xnull}[0]{{\small\textsf{\kw{null}}}}
\newcommand{\filter}[0]{{\small\textsf{\kw{filter}}}}
\newcommand{\filterkey}[0]{{\small\textsf{\kw{filterkey}}}}
\newcommand{\head}[0]{{\small\textsf{\kw{head}}}}

\newcommand{\val}[0]{{\small\textsf{value}}}
\newcommand{\repartition}[0]{{\small\textsf{repartition}}}
\newcommand{\mergeBy}[0]{{\small\textsf{mergeBy}}}

\newcommand{\expl}[1]{\mbox{(#1)}}
\newcommand{\hscode}[1]{{\small\textsf{#1}}}
\newcommand{\datalist}[0]{L}
\newcommand{\aggList}[0]{{\small\textsf{\kw{aggregateList}}}}

\newcommand{\partit}[0]{\mathit{part}}
\newcommand{\cat}[0]{\mathbin{+\!\!\!+}}
\newcommand{\imgof}[1]{\mathit{img}(#1)}
\newcommand{\foldlzof}[1]{{\langle#1\rangle}}
\newcommand{\comboper}[0]{\oplus}
\newcommand{\concat}[0]{{\small\textsf{\kw{concat}}}}
\newcommand{\xs}[0]{\mathit{xs}}
\newcommand{\xss}[0]{\mathit{xss}}
\newcommand{\ys}[0]{\mathit{ys}}

\newcommand{\uscore}[0]{{\char95}}

\newcommand{\attr}[1]{\mathit{attr}({#1})}
\newcommand{\dt}[0]{^D}
\newcommand{\bt}[0]{^T}

\newcommand\numberthis{\addtocounter{equation}{1}\tag{\theequation}}

\pgfplotsset{
    integral axis/.style={
        axis lines=middle,
        enlarge y limits=upper,
        axis equal image, width=8cm,
        xlabel=$x$, ylabel=$y$,
        ytick=\empty,
        xticklabel style={font=\tiny, text height=1.2ex, anchor=north},
        samples=100
    },
    integral/.style={
            domain=2:10,
            samples=6
    },
    integral fill/.style={
            integral,
            draw=none, fill=#1,
            on layer=axis background
        },
        integral fill/.default=cyan!10,
        integral line/.style={
            integral,
            very thick,
            draw=#1
        },
        integral line/.default=black
}

\newcommand{\specspark}{{\sc PureSpark}}

\renewcommand{\vec}[1]{\vv{#1}}

\begin{document}

\title{An Executable Sequential Specification \\ for Spark Aggregation \vspace{-3mm}}


\author{
Yu-Fang Chen\inst{1},
Chih-Duo Hong\inst{1},
Ond\v{r}ej Leng\'{a}l\inst{1,2},\\
Shin-Cheng Mu\inst{1},
Nishant Sinha\inst{3},
Bow-Yaw Wang\inst{1}
}

\institute{
Academia Sinica, Taiwan \and
Brno University of Technology, Czech Republic \and
IBM Research, India
\vspace{-6mm}
}


\maketitle

\begin{abstract}
  Spark is a new promising platform for scalable data-parallel
  computation. It provides several high-level application programming
  interfaces (APIs) to perform parallel data aggregation. 
  Since execution of parallel aggregation in Spark is inherently non-deterministic, a
  natural requirement for Spark programs is to give the
  same result for any execution on the same data set.
  We present {\specspark}, an~executable formal
  Haskell specification  for Spark aggregate combinators.
  Our specification allows us to deduce the precise condition
  for deterministic outcomes from Spark aggregation.
  We
  report~case studies analyzing deterministic outcomes and
  correctness~of Spark programs.
\end{abstract}


\vspace{-7.5mm}
\section{Introduction} \label{section:introduction}
\vspace{-2.0mm}
\enlargethispage{2mm}

\newcommand{\comment}[1]{}

Spark~\cite{spark12,spark-github,sparkCACM} is a~popular platform for scalable distributed data-parallel computation based on a flexible programming environment with concise and high-level APIs.
Spark is by many considered as the successor of MapReduce~\cite{DG10,SLF13}.
Despite its fame, the precursory computational model of MapReduce suffers from I/O congestion and limited programming support for distributed problem solving.
Notably, Spark has the following advantages over MapReduce.
First, it has high performance due to distributed, cached, and in-memory computation. 
Second, the platform adopts a relaxed fault tolerant model where sub-results are recomputed upon faults rather than aggressively stored.
Third, lazy evaluation semantics is used to avoid unnecessary computation. Finally, Spark offers greater programming flexibility through its powerful APIs founded in functional programming. 
Spark also owes its popularity to a unified framework for efficient graph, streaming, and SQL-based relational database computation, a machine learning library, and the support of multiple distributed data storage formats.
Spark is one of the most active open-source projects with over 1000 contributors~\cite{spark-github}.

\hide{
The principal abstraction for data-parallel computation in Spark is called a~Resilient Distributed Dataset (RDD).
An RDD represents a {read-only} collection of data items {partitioned} and stored distributively. The RDD abstraction gives a unified interface to various data storage formats. It also isolates the programmer from the low-level details of reliable distributed computation, such as task management, data distribution, and fault tolerance.
}

In a typical Spark program, a sequence of transformations followed by an action are performed on Resilient Distributed Datasets (RDDs). An RDD is the principal abstraction for data-parallel computation in Spark. It represents a {read-only} collection of data items {partitioned} and stored distributively. RDD operations such as \hscode{map}, \hscode{reduce}, and \hscode{aggregate} are called \emph{combinators}.
They generate and aggregate data in RDDs to carry out Spark computation. 
For instance, the \hscode{aggregate}
combinator takes user-defined functions $\seq$ and $\comb$: $\seq$ accumulates a sub-result for each partition while $\comb$ merges sub-results across different partitions.
Spark also provides a family of {aggregate} combinators for common data structures such as pairs and graphs.
In Spark computation, data aggregation is ubiquitous.

Programming in Spark, however, can be tricky. 
Since sub-results are computed using multiple applications of $\seq$ and $\comb$ across partitions concurrently, the order of their applications varies on different executions.
Because of indefinite orders of computation, aggregation in Spark is inherently {\it non-deterministic}.
A Spark program may produce different outcomes for the same input on different runs. 
This form of non-deterministic computation has other side effects.
For instance, the private function \hscode{AreaUnderCurve.of} in the Spark
machine learning library computes numerical integration distributively; it
exhibits numerical instability due to non-deterministic computation. Consider
the integral of $x^{73}$ on the interval $[-2, 2]$. Since $x^{73}$ is an odd
function, the integral is~$0$. In our experiments, \hscode{AreaUnderCurve.of}
returns different results ranging from $-8192.0$ to $12288.0$ on the same input
because of different orders of floating-point computation.
To ensure deterministic outcomes, programmers must carefully develop their programs to adhere to Spark requirements.

Unfortunately, Spark's documentation does not specify the requirements formally.
It only describes informal algebraic properties about combinators to ensure correctness.
The documentation provides little help to a programmer in understanding the complex, and sometimes unexpected, interaction between $\seq$ and $\comb$,
especially when these two are functions over more complex domains, e.g.\ lists or trees.
Inspecting the Spark implementation is a laborious job since
public combinators are built by composing a long chain of generic private combinators---determining
the execution semantics from the complex implementation is hard. 
Moreover, Spark is continuously evolving and the implementation semantics may change significantly across releases.
We therefore believe that a formal specification of Spark combinators is necessary to help developers understand the program semantics better,
clarify hidden assumptions about RDDs, and help to reason about correctness and sources of non-determinism in Spark programs.


\enlargethispage{2mm}

Building a formal specification for Spark is far from straightforward.
Spark is implemented in Scala and provides high-level APIs also in Python and Java.
Because Spark heavily exploits various language features of Scala, 
it is hard to derive specifications 
without formalizing the operational semantics of the Scala language, which is not an easy task by itself.
Instead of that, we have developed a~Haskell library {\specspark}~\cite{purespark}, which
for each key Spark combinator provides an abstract sequential functional specification in Haskell.
We use Haskell as a specification language for two reasons. 
First, the core of Haskell has strong formal foundations in $\lambda$-calculus. 
Second, program evaluation in Haskell, like in Scala, is lazy, which admits faithful modeling of Spark aggregation. 
Through the use of Haskell we obtain a concise formal functional model for Spark combinators without formalizing Scala.

An important goal of our specification is to make non-determinism in various combinators explicit. 
Spark developers can inspect it to identify sources of non-determinism when program executions yield unexpected outputs.
Researchers can also use it to understand distributed Spark aggregation
and investigate its computational pattern.
Our specification is also {\it executable}. A programmer can use the Haskell APIs to implement data-parallel programs, test them on different input RDDs, and verify correctness of outputs independent of the Spark programming environment.
In our case studies, we capture non-deterministic behaviors of real Spark programs by executing the corresponding {\specspark} specifications with crafted input data sets.
We also show that the sequential specification is useful in developing distributed Spark programs.

Our main contributions are summarized below:
%
\begin{itemize}
\setlength{\itemsep}{0.00mm}
\item We present formal, functional, sequential specifications for key Spark aggregate combinators. The {\specspark} specification consists of executable library APIs. It can assist Spark program development by mimicking data-parallel programming in conventional environments.
\item Based on the specification, 
we investigate and identify necessary and sufficient conditions for Spark aggregate combinators to produce deterministic outcomes for general and pair RDDs.
\item Our specification allows to deduce the precise condition
  for deterministic outcomes from Spark aggregation.
\item We perform a series of case studies on practical Spark programs to validate our formalization.
With {\specspark}, we find instances of numerical instability in the Spark machine learning library.
\item Up to our knowledge, this is the first work to provide a formal, functional specification
of key Spark aggregate combinators for data-parallel computation.
\end{itemize}

\comment{
— Spark platform is an upcoming platform for distributed, data-parallel computation.
(a flexible, scalable distributed compute platform with concise, powerful APIs and higher-order tools)
reasons: in-memory caching of results while simultaneously resilient (avoid intermediate dumps to disk), higher-order functional-programming style API  allows mapping high-level algorithms to  (composition) API calls, concise programs, with low effort. Multiple computation models unified in a single framework — from simple map-reduce computation to graph and stream computation.  Gains from Scala programming patterns (functional, chaining, pattern-matching, anonymous functions, type-omitting?) — moreover exposed as dynamically typed python interfaces. (compare complex class hierarchy based Map-reduce programming in Java)
Thus, Spark is becoming quite popular~\cite{spark-papers-rdd, GTA, loops2seq-comb-synthesis}

They key data structure used by Spark is RDD - abstracts away details of distributed data on clusters while providing familiar, \hscode{map}/\hscode{reduce}/\hscode{aggregate} APIs.
RDD is a collection of lists. Given input list L, transforming it to RDD involves (non-overlapping) partitioning L into n parts, (and then potentially distributing the n parts to different nodes)
Fine-grained ‘combinators’ over RDD for encoding algorithms. Provides 3 main API constructs, defined over two binary operators, seq, comb:
1. map: RDD ‘a -> RDD ‘b  (2) aggregate (`seq, `comb) : RDD * `seq `a -> `comb `b -> RDD `b (3) reduce — seq and comb have same types.
(filter, persist)
Multiple implementations: PairRDD (each partition is list of key-value pairs), DoubleRDD, SequenceFileRDD, VertexRDD, EdgeRDD,

The central combinator is aggregate.  Aggregate is implemented over RDD as follows: apply seq iteratively to each RDD partition p to compute a single value r_p. Now apply comb to each r_p (starting with zero). 
Because the input data may be arbitrarily partitioned, seq and comb may be applied to data in different orders for different runs of the same Spark program.
This may lead to non-deterministic end results, which we want to avoid. 
This brings us to the problem of this paper: when are Spark programs deterministic?

Why should we formalize?
— formal semantics of Spark executions is not known. Embedded in code. Makes it hard to argue for correctness/determinism of Spark programs.
— can help improve design. find more abstractions at the right places.
<example non-deterministic Spark code>: it is easy to assume determinism and not know.
Why do we want determinism? Why should we check for it?
— reproducibility of result
— developer may implicitly assume determinism when encoding the algorithm and forgets to check for it.

This paper:
— provide formal semantics of Spark: formalize various combinators in Spark in a uniform functional programming framework (Haskell). show equivalence with Spark implementation via exhaustive ‘differential’ testing.
— give formal conditions for determinism of these programs.
— give checkers for checking determinism (via reduction to assertion verification of HO haskell programs)
— case studies on a variety of practical Spark programs. Using our formalism to show determinism enables us to find CEs quickly.

Related work:
- Sen, Necula — determinism of concurrent programs.
- Utah prof - differential testing of compilers?
- differential testing: http://www.cse.chalmers.se/~palka/lic/palka_lic_print.pdf

}


%

\vspace{-3.0mm}
\section{Preliminaries} \label{section:preliminaries}
\vspace{-3.0mm}
\enlargethispage{2mm}

Let $A$ be a non-empty set and $\odot : A \times A \to A$ be a~function.
An element $i \in A$ is the \emph{identity} of $\odot$ if for every $a \in A$, it holds that 
$a = i \odot a = a \odot i$.
The function $\odot$ is \emph{associative} if for every $a, a', a'' \in A$,
$a \odot (a' \odot a'') = (a
\odot a') \odot a''$;
$\odot$ is \emph{commutative} if for every $a, a' \in A$, 
$a \odot a' = a' \odot a$.
The algebraic structure $(A, \odot)$ is a~\emph{semigroup} if $\odot$ is associative.
A \emph{monoid} is a~structure $(A, \odot, \bot)$ such that $(A, \odot)$ is a~semigroup and $\bot \in A$ is the identity of $\odot$.
The semigroup $(A, \odot)$ and monoid $(A, \odot, \bot)$ are commutative if $\odot$ is commutative.
\hide{
Observe that for all $A$, the structure $(\listof{A}, \cat, [~])$ is a~monoid.
A~function $h: \listof{A} \to B$ is a~\emph{list homomorphism} to
the monoid $(B, \odot, \bot)$ if
\vspace{-1ex}
\begin{align}
h([~])          &= \bot \\
h(\xs \cat \ys) &= h(\xs) \odot h(\ys).
\end{align}
}

Haskell is a strongly typed purely functional programming
language. Similar to Scala, Haskell programs are lazily
evaluated. 
We use several widely used Haskell functions (Figure~\ref{figure:basic-functions}).
\begin{figure}[t]
  \centering
  \begin{minipage}{.46\linewidth}
\begin{lstlisting}[language=Haskell,basicstyle=\sffamily\footnotesize,morekeywords={reducel,filterkey}]
fst :: ($\alpha$, $\beta$) $\rightarrow$ $\alpha$
fst (x, _) = x

null :: [$\alpha$] $\rightarrow$ Bool
null [] = True
null (x:xs) = False

($\cat$) :: [$\alpha$] $\rightarrow$ [$\alpha$] $\rightarrow$ [$\alpha$]
[] $\cat$ ys = ys
x:xs $\cat$ ys = x:(xs $\cat$ ys)

reducel :: ($\alpha$$\rightarrow$$\alpha$$\rightarrow$$\alpha$)$\rightarrow$[$\alpha$]$\rightarrow$$\alpha$
reducel h (x:xs) = foldl h x xs

concat :: [[$\alpha$]] $\rightarrow$ [$\alpha$]
concat [] = []
concat (xs:xss) = xs $\cat$ (concat xss)

lookup :: $\alpha$ $\rightarrow$ [($\alpha$, $\beta$)] $\rightarrow$ Maybe $\beta$
lookup k [] = Nothing
lookup k ((x, y):xys) = if k == x
       then Just y else lookup k xys

\end{lstlisting}
  \end{minipage}
  \hfill
  \begin{minipage}{.50\linewidth}
\begin{lstlisting}[language=Haskell,basicstyle=\sffamily\footnotesize,morekeywords={reducel,filterkey}]
snd :: ($\alpha$, $\beta$) $\rightarrow$ $\beta$
snd (_, y) = y

elem :: $\alpha$ $\rightarrow$ [$\alpha$] $\rightarrow$ Bool
elem x [] = False
elem x (y:ys) = x==y || elem x ys

map :: ($\alpha$ $\rightarrow$ $\beta$) $\rightarrow$ [$\alpha$] $\rightarrow$ [$\beta$]
map f [] = []
map f (x:xs) = (f x):(map f xs)

foldl :: ($\beta$$\rightarrow$$\alpha$$\rightarrow$$\beta$)$\rightarrow$$\beta$$\rightarrow$[$\alpha$]$\rightarrow$$\beta$
foldl h z [] = z
foldl h z (x:xs) = foldl h (h z x) xs

concatMap :: ($\alpha$ $\rightarrow$ [$\beta$]) $\rightarrow$ [$\alpha$] $\rightarrow$ [$\beta$]
concatMap xs = concat (map f xs)

filter :: ($\alpha$ $\rightarrow$ Bool) $\rightarrow$ [$\alpha$] $\rightarrow$ [$\alpha$]
filter p [] = []
filter p (x:xs) = if p x
     then x:(filter p xs) else filter p xs
\end{lstlisting}
  \end{minipage}
%
  \vspace{-4mm}
  \caption{Basic functions}
  \label{figure:basic-functions}
  \vspace{-4.5mm}
\end{figure}
\hscode{\kw{fst}} and \hscode{\kw{snd}} are projections on pairs. 
\hscode{\kw{null}} tests whether a list is empty. 
\hscode{\kw{elem}} is the membership function for lists; its infix notation
is often used, as in  \hscode{0 `\kw{elem}`~[]}. 
\hscode{($\cat$)} concatenates two lists; it is used as an infix
operator, as in \hscode{[\kw{False}] $\cat$ [\kw{True}]}. 
\hscode{\kw{map}} applies a function to elements of a list.  
\hscode{\kw{reducel}} merges elements of a list by a given binary function
from left to right.  
\hscode{\kw{foldl}} accumulates by applying a function to elements
of a list iteratively, also from left to right.  
\hscode{\kw{concat}} concatenates elements in a list. 
\hscode{\kw{concatMap}} applies a function to elements of a list and
concatenates the results. 
\hscode{\kw{lookup}} finds the value of a key in a list of pairs.
\hscode{\kw{filter}} selects elements from a list by a predicate.

\enlargethispage{2mm}

In order to formalize non-determinism in distributed aggregation, we
define the following non-deterministic shuffle function for lists:
\vspace{-1mm}
\begin{lstlisting}[language=Haskell,basicstyle=\sffamily\footnotesize,morekeywords={shuffle},escapechar=@]
shuffle! :: [$\alpha$] $\rightarrow$ [$\alpha$]
shuffle! xs = ... @~~~~~~~~~@ -- shuffle xs randomly
\end{lstlisting}
\vspace{-1mm}
A~random monad can be used to define random shuffling. Instead
of explicit monadic notation, we introduce the \emph{chaotic}
\hscode{\kw{shuffle}!} function in our presentation for the sake of brevity. Thus,
\hscode{\kw{shuffle}! [0, 1, 2]} evaluates to one of the six possible lists
\hscode{[0, 1, 2]}, \hscode{[0, 2, 1]}, \hscode{[1, 0, 2]}
\hscode{[1, 2, 0]}, \hscode{[2, 0, 1]}, or \hscode{[2, 1, 0]}
randomly. Using
\hscode{\kw{shuffle}!}, more chaotic functions are defined.

\noindent
\begin{minipage}{0.49\linewidth}
\begin{lstlisting}[language=Haskell,basicstyle=\sffamily\footnotesize,morekeywords={shuffle}]
map! :: ($\alpha$ $\rightarrow$ $\beta$) $\rightarrow$ [$\alpha$] $\rightarrow$ [$\beta$]
map! f xs = shuffle! (map f xs)
\end{lstlisting}
\end{minipage}
\begin{minipage}{0.49\linewidth}
\begin{lstlisting}[language=Haskell,basicstyle=\sffamily\footnotesize,morekeywords={shuffle}]
concatMap! :: ($\alpha$ $\rightarrow$ [$\beta$]) $\rightarrow$ [$\alpha$] $\rightarrow$ [$\beta$]
concatMap! f xs = concat (map! f xs)
\end{lstlisting}
\end{minipage}

\noindent
Chaotic \hscode{\kw{map}!} shuffles the result of \hscode{\kw{map}} randomly,
\hscode{\kw{concatMap}!} concatenates the 
shuffled result of \hscode{\kw{map}}. For instance, \hscode{\kw{map}! \kw{even} [0, 1]} 
evaluates to \hscode{[\kw{False}, \kw{True}]} or \hscode{[\kw{True}, \kw{False}]};
\hscode{\kw{concatMap}! \kw{fact} [2, 3]} evaluates to \hscode{[1, 2, 1, 3]} 
or \hscode{[1, 3, 1, 2]} where \hscode{\kw{fact}} computes
a sorted list of factors (note that the two sub-sequences \hscode{[1,2]} and \hscode{[1,3]} are kept intact).
\hide{
e.g., \hscode{fact 6 == [1, 2, 3, 6]}.}

\vspace{-1mm}
\begin{lstlisting}[language=Haskell,basicstyle=\sffamily\footnotesize,morekeywords={shuffle,repartition},escapechar=@]
repartition! :: [$\alpha$] $\rightarrow$ [[$\alpha$]]
repartition! xs = let ys = shuffle! xs ... 
                $\,$in yss  @~~~~~~~~~@ -- ys == concat yss
\end{lstlisting}
\vspace{-1mm}
The function \hscode{\kw{repartition}!}
shuffles a given list and partitions the shuffled list into several
non-empty lists.
For instance, \hscode{\kw{repartition}! [0, 1]} results in \hscode{[[0], [1]]},
\hscode{[[1], [0]]}, \hscode{[[0, 1]]}, or \hscode{[[1, 0]]}.
The chaotic function can be implemented by a~random monad easily;
its precise definition is omitted here.


\vspace{-3.0mm}
\section{Spark Aggregation} \label{section:data-parallel-computation}
\vspace{-2.0mm}
\emph{Resilient Distributed Datasets (RDDs)} are the basic data abstraction
in Spark. An RDD is a collection of partitions of  
immutable data; data in different partitions can be processed concurrently.
We formalize partitions by lists, and RDDs by
lists of partitions.

\noindent
\begin{minipage}{0.49\linewidth}
\begin{lstlisting}[language=Haskell,basicstyle=\sffamily\footnotesize]
type Partition $\alpha$ = [$\alpha$]
\end{lstlisting}
\end{minipage}
\begin{minipage}{0.49\linewidth}
\begin{lstlisting}[language=Haskell,basicstyle=\sffamily\footnotesize]
type RDD $\alpha$ = [Partition $\alpha$]
\end{lstlisting}
\end{minipage}

The Spark \hscode{aggregate} combinator computes \emph{sub-results} of every
partitions in an RDD, and returns the aggregated result by combining
sub-results.
\vspace{-1mm}
\begin{lstlisting}[language=Haskell,basicstyle=\sffamily\footnotesize,deletekeywords={seq},morekeywords={aggregate}]
aggregate :: $\beta$ $\rightarrow$ ($\beta$ $\rightarrow$ $\alpha$ $\rightarrow$ $\beta$) $\rightarrow$ ($\beta$ $\rightarrow$ $\beta$ $\rightarrow$ $\beta$) $\rightarrow$ RDD $\alpha$ $\rightarrow$ $\beta$
aggregate z seq comb rdd = let presults = map! (foldl seq z) rdd
                           $\,$in foldl comb z presults
\end{lstlisting}
\vspace{-1mm}
More concretely, let \hscode{z} be a default aggregated value. 
\hscode{\kw{aggregate}} applies \hscode{\kw{foldl} seq z} to every
partition of \hscode{rdd}. Hence the sub-result of each partition is
accumulated by folding elements in the partition with 
\hscode{seq}. The combinator then combines
sub-results by another folding using \hscode{comb}.

Note that the chaotic \hscode{\kw{map}!} function is used to model
non-deterministic interleavings of sub-results.
To exploit concurrency, Spark creates a task to
compute the sub-result for each partition. These tasks are executed
concurrently and hence
induce non-deterministic computation. We use the chaotic \hscode{\kw{map}!}
function to designate non-de\-ter\-min\-is\-m explicitly. 

A related combinator is \hscode{\kw{reduce}}. 
Instead of \hscode{\kw{foldl}}, the combinator uses
\hscode{\kw{reducel}} to aggregate data in an RDD.
\vspace{-1mm}
\begin{lstlisting}[language=Haskell,basicstyle=\sffamily\footnotesize,morekeywords={reducel,reduce}]
reduce :: ($\alpha$ $\rightarrow$ $\alpha$ $\rightarrow$ $\alpha$) $\rightarrow$ RDD $\alpha$ $\rightarrow$ $\alpha$
reduce comb rdd = let presults = map! (reducel comb) rdd
                   in reducel comb presults
\end{lstlisting}
\vspace{-1mm}
Similar to the \hscode{\kw{aggregate}} combinator, \hscode{\kw{reduce}} computes
sub-results concurrently. The chaotic \hscode{\kw{map}!} function
is again used to model non-deterministic computation.

Sub-results of different partitions are computed in parallel,
but the \hscode{\kw{aggregate}} combinator still combines sub-results
sequentially. 
This can be further parallelized.
Observe that several sub-results may be available simultaneously
from distributed computation. 
The Spark \hscode{\kw{treeAggregate}} combinator 
applies \hscode{comb} to pairs of sub-results concurrently until
the final result is obtained.
In addition to concurrent computation of sub-results, \hscode{\kw{treeAggregate}}
also combines sub-results from different partitions in
parallel. 

In our specification, two chaotic functions are used to model
non-deterministic computation on two different levels. The
\hscode{\kw{map}!} function models non-determinism in computing
sub-results of partitions. The \hscode{\kw{apply}!} function (introduced below) models
concurrent combination of sub-results from different partitions. It 
combines two consecutive sub-results picked chaotically, and 
repeats such chaotic combinations until
the final result is obtained. Observe 
that the computation has a binary-tree structure
with \hscode{comb} as internal nodes and sub-results from different partitions as leaves. 
\vspace{-1mm}
\begin{lstlisting}[language=Haskell,basicstyle=\sffamily\footnotesize,morekeywords={divide,apply,treeAggregate},deletekeywords={seq},escapechar=@]
apply! :: ($\beta$ $\rightarrow$ $\beta$ $\rightarrow$ $\beta$) $\rightarrow$ [$\beta$] $\rightarrow$ $\beta$
apply! comb [r] = r
apply! comb [r, r'] = comb r r'
apply! comb rs = let (ls', l', r', rs') = ... @~~~~~~~~~@ -- rs == ls' $\cat$ [l', r'] $\cat$ rs'
  in apply! comb (ls' $\cat$ [comb l' r'] $\cat$ rs')

treeAggregate:: $\beta$ $\rightarrow$ ($\beta$$\rightarrow$$\alpha$$\rightarrow$$\beta$) $\rightarrow$ ($\beta$$\rightarrow$$\beta$$\rightarrow$$\beta$) $\rightarrow$ RDD $\alpha$ $\rightarrow$ $\beta$
treeAggregate z seq comb rdd = let presults = map! (foldl seq z) rdd
  in apply! comb presults
\end{lstlisting}
\vspace{-1mm}

The \hscode{\kw{treeReduce}} combinator optimizes \hscode{\kw{reduce}} by
combining sub-results in parallel. 
Similar to \hscode{\kw{treeAggregate}}, two levels of
non-deterministic computation can occur.
\vspace{-1mm}
\begin{lstlisting}[language=Haskell,basicstyle=\sffamily\footnotesize,morekeywords={reducel,treeReduce,apply}]
treeReduce :: ($\alpha$ $\rightarrow$ $\alpha$ $\rightarrow$ $\alpha$) $\rightarrow$ RDD $\alpha$ $\rightarrow$ $\alpha$
treeReduce comb rdd = let presults = map! (reducel comb) rdd
                       in apply! comb presults
\end{lstlisting}
\vspace{-1mm}

\vspace{-3mm}
\paragraph{Pair RDDs.}

Key-value pairs are widely used in data parallel computation. If the data type
of an RDD is a~pair, we say that the 
RDD is a~\emph{pair} RDD. The first and second elements in a pair are
called the \emph{key} and the \emph{value} of the pair respectively.
\vspace{-1mm}
\begin{lstlisting}[language=Haskell,basicstyle=\sffamily\footnotesize]
type PairRDD $\alpha$ $\beta$ = RDD ($\alpha$, $\beta$)
\end{lstlisting}
\vspace{-1mm}
In a pair RDD, different pairs can have the same key. Spark provides
combinators to aggregate values associated with the same key. The
\hscode{\kw{aggregateByKey}} combinator returns an RDD 
by aggregating values associated with the same key.
We use the following functions to formalize \hscode{\kw{aggregateByKey}}:

\noindent
\begin{minipage}{0.49\linewidth}
\begin{lstlisting}[language=Haskell,basicstyle=\sffamily\footnotesize,morekeywords={hasKey}]
hasKey :: $\alpha$ $\rightarrow$ Partition ($\alpha$, $\beta$) $\rightarrow$ Bool
hasKey k ps = case (lookup k ps) of
  Just _  $\rightarrow$ True   
  Nothing $\rightarrow$ False
\end{lstlisting}
\end{minipage}
\begin{minipage}{0.49\linewidth}
\begin{lstlisting}[language=Haskell,basicstyle=\sffamily\footnotesize,morekeywords={hasValue}]
hasValue :: $\alpha$ $\rightarrow$ $\beta$ $\rightarrow$ Partition ($\alpha$, $\beta$) $\rightarrow$ $\beta$
hasValue k val ps = case (lookup k ps) of
  Just v $\rightarrow$ v
  Nothing $\rightarrow$ val
\end{lstlisting}
\end{minipage}

\vspace{-1mm}
\begin{lstlisting}[language=Haskell,basicstyle=\sffamily\footnotesize,morekeywords={addTo}]
addTo :: $\alpha$ $\rightarrow$ $\beta$ $\rightarrow$ Partition ($\alpha$, $\beta$) $\rightarrow$ Partition ($\alpha$, $\beta$)
addTo key val ps = foldl ($\lambda$r (k, v) $\rightarrow$ if key == k then r else (k, v):r) [(key, val)] ps
\end{lstlisting}

\noindent
The expression \hscode{\kw{hasKey} k ps} checks if
\hscode{key} appears in a partition of pairs.
\hscode{\kw{hasValue} k val ps} finds a value associated
with \hscode{key} in a partition of pairs. It evaluates to the
default value \hscode{val} if \hscode{key} does not appear in the
partition.
The expression \hscode{\kw{addTo} key val ps} adds the pair
\hscode{(key, val)} to the partition \hscode{ps}, and removes other pairs
with the same key. 

The \hscode{\kw{aggregateByKey}} combinator
first aggregates all pairs with the value \hscode{z} and the function
\hscode{mergeComb} in each partition. If values \hscode{vs} are
associated with the same key in a partition, the value
\hscode{\kw{foldl} mergeComb z vs} for the key is pre-aggregated. 
Since a~key may
appear in several partitions, all pre-aggregated values associated
with the key  across different partitions are merged
using \hscode{mergeValue}.
\begin{lstlisting}[language=Haskell,basicstyle=\sffamily\footnotesize,morekeywords={repartition,aggregateByKey,addTo}]
aggregateByKey :: $\gamma$ $\rightarrow$ ($\gamma$ $\rightarrow$ $\beta$ $\rightarrow$ $\gamma$) $\rightarrow$ ($\gamma$ $\rightarrow$ $\gamma$ $\rightarrow$ $\gamma$) $\rightarrow$ PairRDD $\alpha$ $\beta$ $\rightarrow$ PairRDD $\alpha$ $\gamma$
aggregateByKey z mergeComb mergeValue pairRdd =
  let mergeBy fun left (k, v) = addTo k (fun (hasValue k z left) v) left
     preAgg = concatMap! (foldl (mergeBy mergeComb) []) pairRdd
  in repartition! (foldl (mergeBy mergeValue) [] preAgg)
\end{lstlisting}
In the specification, we accumulate values associated with the same key by
\hscode{mergeComb} in each partition, keeping a~list of pairs of a~key and the partially aggregated
value for the key.
Since accumulation in different
partitions runs in parallel, the chaotic \hscode{\kw{concatMap}!} function
is used to model such non-deterministic computation. After all
partitions finish their accumulation, \hscode{mergeValue} merges
values associated with the same key across different partitions. The
final pair RDD can have a default or user-defined partitioning.
Since a user-defined partitioning may shuffle a pair RDD
arbitrarily, it is in our specification modeled by the chaotic \hscode{\kw{repartition}!}
function.

\enlargethispage{2mm}

Pair RDDs have a combinator corresponding to \hscode{\kw{reduce}} called
\hscode{\kw{reduceByKey}}. \hscode{\kw{reduceByKey}} merges
all values associated with a key by \hscode{mergeValue}, following a
similar computational pattern as
\hscode{\kw{aggregateByKey}}. Note that every key is associated with at most
one value in resultant pair RDDs of
\hscode{\kw{aggregateByKey}} or \hscode{\kw{reduceByKey}}.
\vspace{-1mm}
\begin{lstlisting}[language=Haskell,basicstyle=\sffamily\footnotesize,morekeywords={repartition,reduceByKey,addTo}]
reduceByKey :: ($\beta$ $\rightarrow$ $\beta$ $\rightarrow$ $\beta$) $\rightarrow$ PairRDD $\alpha$ $\beta$ $\rightarrow$ PairRDD $\alpha$ $\beta$
reduceByKey mergeValue pairRdd =
  let merge left (k, v) = case lookup k left of Just v' $\rightarrow$ addTo k (mergeValue v' v) left
                                             Nothing $\rightarrow$ addTo k v left
     preAgg = concatMap! (foldl merge []) pairRdd
  in repartition! (foldl merge [] preAgg)
\end{lstlisting}
\vspace{-1mm}

\noindent
Spark also provides a library, called GraphX, for a~distributed analysis of graphs.
See App.~\ref{app:graph_rdds} for a~formalization of some of its key functions.

\vspace{-3.0mm}
\section{Deterministic Aggregation} \label{section:deterministic-aggregation}
\vspace{-2.0mm}

\newcommand{\propose}[1]{{\color{blue}New proposal: #1}}
\newcommand{\anotherpropose}[1]{{\color{blue}Yet new proposal: #1}}

Having deterministic outcomes is desired from all aggregation functions.
If a~function may return different values on different executions, the function is often not implemented correctly.
A program with explicit assumptions on the input data is also desirable.
Otherwise, the program may work correctly on certain data sets
but produce unexpected outcomes on others where implicit assumptions do not hold~\cite{xiao14mr}.
We now investigate conditions under which Spark aggregation
combinators always produce deterministic outcomes.
Proofs of the given lemmas can be found in App.~\ref{app:proofs}.
Proofs of some crucial lemmas have also been formalized using Agda~\cite{purespark}.


We first show how to deal with non-deterministic behaviors
in the \hscode{\kw{aggregate}} combinator.
Consider a variant of the formalization of
\hscode{\kw{aggregate}} from
Section~\ref{section:data-parallel-computation}:
\vspace{-1mm}
\begin{lstlisting}[language=Haskell,basicstyle=\sffamily\footnotesize,deletekeywords={seq},morekeywords={aggregate}]
aggregate'::$\beta$ $\rightarrow$ ($\beta$ $\rightarrow$ $\alpha$ $\rightarrow$ $\beta$) $\rightarrow$ ($\beta$ $\rightarrow$ $\beta$ $\rightarrow$ $\beta$) $\rightarrow$ RDD $\alpha$ $\rightarrow$ $\beta$
aggregate' z seq comb rdd = let presults = perm (map (foldl seq z) rdd)
  in foldl comb z presults
\end{lstlisting}
\vspace{-1mm}
Observe that we changed the application of the chaotic \hscode{\kw{map}!} function
with an application of the permutation \hscode{perm} after the regular \hscode{\kw{map}} function.
The function composition \hscode{perm( \kw{map} ...)} is a~concrete instantiation
of \hscode{\kw{map}!}, that is, a~function that permutes its list argument.
%
Notice that \hscode{perm} can be pushed inside \hscode{\kw{map}}:
%
\begin{lstlisting}[language=Haskell,basicstyle=\sffamily\footnotesize,morekeywords={permutations}]
perm (map f xs) == map f (perm xs).
\end{lstlisting}
\vspace{-1mm}
Assume that \hscode{rdd} was obtained from a~list~\hscode{xs} by splitting and permuting, that is,
\hscode{rdd == perm' (split xs)} where \hscode{split :: [$\alpha$] $\rightarrow$ [[$\alpha$]]} satisfies \hscode{xs == (concat . split) xs}.
We can therefore rewrite the computation of \hscode{presults} in \hscode{\kw{aggregate}'} to
\vspace{-1mm}
\begin{lstlisting}[language=Haskell,basicstyle=\sffamily\footnotesize,deletekeywords={seq, split}]
let pres = perm (map (foldl seq z) (perm' (split xs))),
\end{lstlisting}
\vspace{-1mm}
After pushing \hscode{perm} inside \hscode{map}, we obtain
\vspace{-1mm}
\begin{lstlisting}[language=Haskell,basicstyle=\sffamily\footnotesize,deletekeywords={seq, split}]
let pres = map (foldl seq z) ((perm . perm') (split xs)).
\end{lstlisting}
\vspace{-1mm}
Since \hscode{perm . perm'} is also a permutation \hscode{perm''}, we have
\vspace{-1mm}
\begin{lstlisting}[language=Haskell,basicstyle=\sffamily\footnotesize,deletekeywords={seq, split}]
let pres = map (foldl seq z) rdd'
\end{lstlisting}
\vspace{-1mm}
where \hscode{rdd'} is another RDD obtained from \hscode{xs} by splitting and shuffling.
Let us call (deterministic) instances of \hscode{\kw{repartition}!} as \emph{partitionings}. 
As a~consequence, we focus only on proving if calls to
\hscode{\kw{aggregate}$\dt$} defined below have deterministic outcomes
for different partitionings of a list into RDDs:
\vspace{-1mm}
\begin{lstlisting}[language=Haskell,basicstyle=\sffamily\footnotesize,deletekeywords={seq},morekeywords={aggregate}]
aggregate$\dt$:: $\beta$ $\rightarrow$ ($\beta$ $\rightarrow$ $\alpha$ $\rightarrow$ $\beta$) $\rightarrow$ ($\beta$ $\rightarrow$ $\beta$ $\rightarrow$ $\beta$) $\rightarrow$ RDD $\alpha$ $\rightarrow$ $\beta$
aggregate$\dt$ z seq comb rdd = let pres = map (foldl seq z) rdd
                            $\ $in foldl comb z pres
\end{lstlisting}
\vspace{-1mm}
Moreover, we define deterministic versions of \hscode{\kw{reduce}}
\vspace{-1mm}
\begin{lstlisting}[language=Haskell,basicstyle=\sffamily\footnotesize,morekeywords={reducel},morekeywords={reduce}]
reduce$\dt$ :: ($\alpha$ $\rightarrow$ $\alpha$ $\rightarrow$ $\alpha$) $\rightarrow$ RDD $\alpha$ $\rightarrow$ $\alpha$
reduce$\dt$ comb rdd = let presults = perm (map (reducel comb) rdd)
  in reducel comb presults
\end{lstlisting}
\vspace{-1mm}
and also \hscode{\kw{treeAggregate}$\dt$} and \hscode{\kw{treeReduce}$\dt$} in a~similar way.

In the following, given a~function \hscode{f} that takes an RDD as one of its
parameters and contains a~single occurrence of the chaotic \hscode{\kw{map}!}
(respectively~\hscode{\kw{concatMap}!}) function, we use \hscode{f$\dt$} to denote the
function obtained from \hscode{f} by replacing the chaotic \hscode{\kw{map}!}
(respectively~\hscode{\kw{concatMap}!}) with a~regular \hscode{\kw{map}}
(respectively~\hscode{\kw{concatMap}}).
A~similar reasoning 
can show that it suffices
to check whether calls to \hscode{f$\dt$} have deterministic outcomes for
different partitionings on a list into RDDs.

For better readability, standard mathematical notation
of functions is used in the rest of this section. We represent a~Haskell function application \hscode{f
x1 \dots\:
xn} as $f(x_1, \dots, x_n)$.

\vspace{-3.0mm}
\subsection{\textsf{aggregate}}\label{sec:aggregate}
\vspace{-2.0mm}

In this section, we give conditions for deterministic outcomes of calls to the
aggregate combinator $\kw{\agg}(\zero, \seq, \allowbreak \comboper, \rdd)$ for 
$\zero :: \beta$, $\seq :: \beta \times \alpha \to \beta$, 
$\comboper :: \beta \times \beta \to \beta$, and $\rdd :: \RDD\ \alpha$.
We first define what it means for calls to the $\agg$ combinator to have
deterministic outcomes.
\vspace{-1mm}
\begin{definition}\label{def:agg-det-out}
Calls to
$\kw{\agg}(\zero, \seq, \comboper, \rdd)$ have \emph{deterministic outcomes}
if
\vspace{-1mm}
\begin{equation}
  \kw{\agg}\dt(\zero, \seq, \comboper, \partit(\datalist)) = \kw{\foldl}(\seq, \zero, \datalist)
\vspace{-1mm}
\end{equation}
for all lists $\datalist$ and partitionings $\partit$.
\end{definition}

\enlargethispage{1mm}

Conventionally, $\agg$ is regarded as a parallelized counterpart of $\foldl$.
For example, the sequential $\agg$ function in the standard Scala library
ignores the $\comboper$ operator and is implemented by $\foldl$.
This is why we characterize deterministic $\kw{\agg}$ as $\foldl$ in Definition~\ref{def:agg-det-out}.
Our characterization, however, does not cover all $\kw{\agg}$ calls that always give the same outputs.
In particular, it does not cover an $\kw{\agg}$ call where $\comboper$ is a~constant function,
which is, however, quite suspicious in a~distributed data-parallel computation
and should be reported.

We give necessary and sufficient conditions for $\kw{\agg}$ calls to have
deterministic outcomes in several lemmas, culminating in Corollary~\ref{col:det-agg}.
The first lemma allows us to check only conditions on $\seq$ and $\comboper$
over all possible pairs of lists instead of enumerating all possible
partitionings on lists.
For brevity, we use $\foldlzof{p_1}$ for $\kw{\foldl}(\seq, \zero, p_1)$,
and
$\imgof{\foldl(\seq, \zero)}$ for the image of
$\foldl(\seq, \zero, \datalist)$ for any list $\datalist$.
That is, $\imgof{\foldl(\seq, \zero)} = \{y \mid \textmd{there is a list\,} 
\datalist \textmd{ such that } \foldl(\seq, \zero, \datalist) = y\}$.

\begin{restatable}{lemma}{homoIff}\label{lem:homoIff}
Calls to $\agg(\zero, \seq, \comboper, \rdd)$ have deterministic outcomes iff:
\vspace{-1ex}
\begin{enumerate}
  \item  $(\imgof{\foldl(\seq, \zero)}, \comboper, \zero)$ is a~commutative monoid, and
  \item for all lists $p_1, p_2 :: \listof{\alpha}, \foldlzof{p_1 \cat p_2} = \foldlzof{p_1} \comboper \foldlzof{p_2}$ .
\end{enumerate}
\end{restatable}

Note that condition 2 in Lemma~\ref{lem:homoIff} is equivalent to saying that
$\foldlzof{\cdot}$ is a~list
homomorphism to the monoid $(\imgof{\foldl(\seq, \zero)}, \comboper, \zero)$~\cite{B:TL:86}.

The lemma below further helps us reduce the need of testing conditions over all
possible pairs of lists to conditions over elements of $\alpha \times \imgof{\foldl(\seq, \zero)}$.

\begin{restatable}{lemma}{aggreSuffCond}
  %
  Let $\comboper$ be associative on $\gamma = \imgof{\foldl(\seq,
  \zero)}$ and $\zero$ be the identity of $\comboper$ on $\gamma$.
  The following are equivalent:
  \begin{enumerate}

  \item  for all lists $p_1, p_2 :: \listof{\alpha}$,
    \vspace{-1ex}
    \begin{equation}\label{eq:9461f388-4968-41b8-83af-b8a49f9cfedc}
      \foldlzof{p_1 \cat p_2} = \foldlzof{p_1} \comboper \foldlzof{p_2},
    \end{equation}

  \item  for all elements $d :: \alpha$ and $e :: \gamma$,
    \vspace{-1ex}
    \begin{equation}\label{eq:cd2c544e-7b22-4c31-b1d0-1160239694b7}
      \seq(e, d) = e \comboper \seq(\zero, d).
    \end{equation}
  \end{enumerate}
\end{restatable}

\noindent Summarizing the lemmas, we get the following corollary:

\begin{corollary}\label{col:det-agg}
  Calls to $\agg(\zero, \seq, \comboper, \rdd)$ have deterministic outcomes
  iff
  \vspace{-1ex}
  \begin{enumerate}
    \item  $(\imgof{\foldl(\seq, \zero)}, \comboper, \zero)$ is a~commutative
      monoid and
    \item  for all $d :: \alpha$ and $e :: \imgof{\foldl(\seq, \zero)}$, it holds that $\seq(e, d)
      = e \comboper \seq(\zero, d)$.
  \end{enumerate}
\end{corollary}

\vspace{-3.0mm}
\subsection{\textsf{reduce}}\label{sec:reduce}
\vspace{-2.0mm}

This section explores conditions for deterministic outcomes of 
calls to $\red(\comboper, \rdd)$ for $\comboper :: \alpha \times \alpha \to \alpha$ and $\rdd :: \RDD\ \alpha$.
We use the function \hscode{\kw{reduce}$\dt$} defined in the
introduction of Section~\ref{section:deterministic-aggregation}.
For \hscode{\kw{reduce}}, we assume that for any non-empty list, all partitions
of its~partitioning are non-empty (otherwise the result of \hscode{\kw{reduce}} is
undefined).

\enlargethispage{2mm}

We define deterministic outcomes for $\red$ as follows.
\begin{definition}
Calls to
$\red(\comboper, \rdd)$ have \emph{deterministic outcomes}
if
%
%
\vspace{-1ex}
\begin{equation}
  \red\dt(\comboper, \partit(\datalist)) = \redl(\comboper, \datalist)
\end{equation}
for all lists $\datalist$ and partitionings $\partit$.
\end{definition}

\newsavebox{\redseqcombbox}
\begin{lrbox}{\redseqcombbox}
\noindent
\hspace*{-6mm}
\begin{minipage}{0.39\linewidth}
\begin{lstlisting}[language=Haskell,basicstyle=\sffamily\footnotesize, morekeywords={aggregate}, deletekeywords={seq}, mathescape]
seq' x y = case x of
  Nothing $\rightarrow$ Just y
  Just x' $\rightarrow$ Just (x' $\comboper$ y)
\end{lstlisting}
\end{minipage}
\begin{minipage}{0.59\linewidth}
\begin{lstlisting}[language=Haskell,basicstyle=\sffamily\footnotesize, morekeywords={aggregate}, deletekeywords={seq}, mathescape]
($\comboper$') x y = case (x, y) of (Nothing, y')      $\rightarrow$ y'
            (x', Nothing)      $\rightarrow$ x'
            (Just x', Just y') $\rightarrow$ Just (x' $\comboper$ y') .
\end{lstlisting}
\end{minipage}
\end{lrbox}
\vspace{-1mm}

We reduce the problem of checking if $\red$ has deterministic outcomes to the problem of checking if $\agg$ has deterministic outcomes by the following lemma.

\vspace{-2mm}
\begin{restatable}{lemma}{lemmaRedIsAggr}\label{lem:lemmaRedIsAggr}
Calls to $\red(\comboper, \rdd)$ have deterministic
outcomes iff calls to $\agg(\allowbreak \nothing, \seq', \comboper', \rdd)$ have deterministic outcomes, where $\seq'$ and $\comboper'$ are as follows:

\vspace{0mm}
\usebox{\redseqcombbox}
\end{restatable}
\vspace{-3mm}
Combining Corollary~\ref{col:det-agg} and Lemma~\ref{lem:lemmaRedIsAggr}, we
get the condition for deterministic outcomes of $\red(\comboper,
\rdd)$ calls.
\begin{restatable}{corollary}{coroReduce}\label{col:det-red}
Calls to $\red(\comboper, \rdd)$ have deterministic outcomes iff
$(\alpha, \comboper)$ is a commutative semigroup.
\end{restatable}

\vspace{-4.0mm}
\subsection{\textsf{treeAggregate} and \textsf{treeReduce}}\label{sec:treeAggregate}
\vspace{-2.0mm}

\enlargethispage{3mm}

This section gives conditions for deterministic outcomes of calls to the following
two aggregate combinators:
\begin{enumerate}
  \item  $\treeagg(\zero, \seq, \comboper, \rdd)$ for 
    $\zero :: \beta$, $\seq :: \beta \times \alpha \to \beta$, 
    $\comboper :: \beta \times \beta \to \beta$,
    and $\rdd :: \RDD\ \alpha$; and

  \item  $\treered(\comboper, \rdd)$ for $\comboper :: \alpha \times \alpha \to \alpha$, $\rdd :: \RDD\ \alpha$.
\end{enumerate}
Different from \hscode{\kw{aggregate}} and \hscode{\kw{reduce}}, the tree variants have
another level of non-de\-ter\-min\-is\-m modeled by
\hscode{\kw{apply}!}. The chaotic function effectively 
simulates non-deterministic computation with a binary-tree structure (Section~\ref{section:data-parallel-computation}).

To define calls to \hscode{\kw{treeAggregate}} and
\hscode{\kw{treeReduce}} to have deterministic outcomes, we use the functions
\hscode{\kw{treeAggregate}$\bt$} and \hscode{\kw{treeReduce}$\bt$} 
obtained by 
adding an explicit deterministic instantiation of \hscode{\kw{apply}!}
to \hscode{\kw{treeAggregate}$\dt$} and \hscode{\kw{treeReduce}$\dt$}.
%
\begin{definition}
Calls to
$\treeagg(\zero, \seq, \comboper, \rdd)$ and $\treered(\comboper,
\rdd)$ have \emph{deterministic outcomes} if
\vspace{-1ex}
\begin{equation}
  \treeagg\bt\!(\apply,\zero, \seq, \comboper, \partit(\datalist)) = \foldl(\seq, \zero, \datalist)
\end{equation}
and
\vspace{-1ex}
\begin{equation}
  \treered\bt\!(\apply,\comboper, \partit(\datalist)) = \redl(\comboper, \datalist)
\end{equation}
respectively for all lists $\datalist$, partitionings
$\partit$, and instantiations $\apply$ of \hscode{\kw{apply}!}.
\end{definition}

The following two propositions state necessary and sufficient conditions for the $\treeagg$ and $\treered$ combinators to have deterministic outcomes.
\hide{
The proof of these propositions are more difficult mainly due to the chaotic \hscode{divide!} function that randomly picks and combines two consecutive intermediate results.
}
\vspace{-1mm}
\begin{restatable}{proposition}{lemmaTreeAgg}\label{lemma:treeAgg}
Calls to $\treeagg(\zero, \seq, \comboper, \rdd)$ have deterministic outcomes
iff calls to $\agg(\zero, \seq, \comboper, \rdd)$ have deterministic outcomes.
\end{restatable}

\enlargethispage{4mm}

\vspace{-1mm}
\begin{restatable}{proposition}{lemTreeRedIsRed}
Calls to $\treered(\comboper, \rdd)$ have deterministic outcomes iff
calls to $\red(\comboper, \rdd)$ have deterministic outcomes.
\end{restatable}

\vspace{-3.0mm}
\subsection{\textsf{aggregateByKey} and \textsf{reduceByKey}}\label{sec:aggregateByKey}
\vspace{-1.0mm}

We proceed by investigating conditions for the following combinators on pair RDDs:
\begin{enumerate}
  \item  $\aggBKey(\zero, \seq, \comboper, \prdd)$ for $\zero :: \gamma$, $\seq :: \gamma \times \beta
    \to \gamma$, $\comboper :: \gamma \times \gamma \to \gamma$, and $\prdd ::
    \PairRDD\ \alpha\ \beta$; and

  \item  $\redBKey(\comboper, \prdd)$ for $\comboper ::
    \beta \times \beta \to \beta$ and $\prdd :: \PairRDD\ \alpha\ \beta$.
\end{enumerate}

\noindent
We define an auxiliary function $\filterkey$ that obtains
a~list of all values associated with the given key from a list of pairs.
\vspace{-1mm}
\begin{lstlisting}[language=Haskell,basicstyle=\sffamily\footnotesize,morekeywords={filterkey}]
filterkey :: $\alpha$ $\rightarrow$ [($\alpha$, $\beta$)] $\rightarrow$ [$\beta$]
filterkey _ [] = []
filterkey k (k, v):xs = v:(filterkey k xs)
filterkey k (_, _):xs = filterkey k xs
\end{lstlisting}
\vspace{-1mm}

\noindent Deterministic outcomes of calls to $\aggBKey$ are now defined using
the function $\aggBKey\dt$ as follows.
\begin{definition}
Calls to
$\aggBKey(\zero, \seq, \comboper, \prdd)$ have
\emph{deterministic outcomes} if 
\vspace{-1ex}
\begin{equation*}
  \lookup(k, \aggBKey\dt(\zero, \seq, \comboper, \partit(\datalist)))
  = \foldl(\zero, \seq, \filterkey(k, \datalist))
\end{equation*}
for all lists $\datalist$ of pairs, partitionings $\partit$, and keys~$k$.
\end{definition}

\hide{
The following lemma helps us to change the analysis of the conditions when
$\aggBKey$ has a deterministic outcome to the analysis of $\aggWKey$, which
handles the aggregation with respect to each key individually.

\begin{restatable}{lemma}{lemAggBKeyIsAggWKey}\label{lem:lemAggBKeyIsAggWKey}
It holds that
\vspace{-1ex}
\begin{align*}
  \MoveEqLeft
  \lookUp(k, \aggBKey(\zero, \seq, \comboper, \prdd)) \\
  &= \aggWKey(k, \zero, \seq, \comboper, \prdd)) ,
\end{align*}
\end{restatable}
\noindent where $\lookUp$ searches the first value with a given
key in an RDD:
\begin{align*}
& \lookUp(k, \xss) = \\
&\quad \quad \head_z(\concat(\map(\map(\val \circ \filterkey~k),
\xss))) ,
\end{align*}
and $\head_z$ returns $z$ when the input is empty.
}

Finally, the following proposition states the conditions that need to hold for
calls to $\aggBKey$ to have deterministic outcomes.
\begin{restatable}{proposition}{lemAggByKeyIsAgg}\label{lem:lemAggByKeyIsAgg}
  Calls to $\aggBKey(\zero, \seq, \comboper, \prdd)$ have
  deterministic outcomes iff calls to
  $\agg(\zero, \seq, \comboper, \rdd)$ have deterministic
  outcomes.
\end{restatable}

\noindent We define when calls to $\redBKey$ have deterministic outcomes via $\redBKey\dt$.
\begin{definition}
Calls to $\redBKey(\comboper, \prdd)$ have
\emph{deterministic outcomes} if
\vspace{-1ex}
\begin{equation*}
  \lookup(k, \redBKey\dt(\comboper, \partit(\datalist))) = \redl(\comboper, \filterkey(k, \datalist))
\end{equation*}
for all list $\datalist$ of pairs, partitioning $\partit$, and key $k$.
\end{definition}

\begin{restatable}{proposition}{lemRedBKeyIsRed}\label{lem:lemRedBKeyIsRed}
  Calls to $\redBKey(\comboper, \prdd)$ have deterministic outcomes iff
  calls to $\red(\comboper, \rdd)$ have deterministic outcomes.
\end{restatable}

\vspace{-3.0mm}
\subsection{Discussion}\label{sec:testing-det-aggr}
\vspace{-2.0mm}

Our conditions for deterministic outcomes are more general than it
appears. In addition to scalar data, such as integers, 
they are also applicable to RDDs containing non-scalar data,
such as lists or sets. 
In our extended set of
case studies, we will prove deterministic outcomes from a
distributed Spark program using non-scalar data
(App.~\ref{app:case_studies}). 

Corollary~\ref{col:det-agg} gives necessary and sufficient conditions for calls to $\agg$ to have deterministic outcomes.
Instead of checking whether $\agg$ computes the same result on
all possible partitionings on any list for given $\zero$, $\seq$, 
and $\comb$, 
the corollary, instead, allows us to investigate properties for all elements of
$\imgof{\foldl(\seq, \zero)} \times \imgof{\foldl(\seq, \zero)}$
and 
$\alpha \times \imgof{\foldl(\seq, \zero)}$.
Our precise conditions reduce the need of checking all partitionings
to checking all elements of Cartesian products. 
It appears that deterministic outcomes from calls to combinators
can be verified automatically. The problem, however, remains difficult
for the following reasons:

\hide{
Note that $\imgof{\foldl(\seq, \zero)}\subseteq \beta$. These conditions remove one source of unboundedness when testing if an $\agg$ combinator has deterministic outcomes in the sense that they reduce the need of checking all lists in $\listof{\listof{\alpha}}$ to checking all elements in a subset of $(\alpha\cup\beta) \times \beta$.
Even with this reduction, if we do not put any restrictions on $\seq$ and $\comb$, the task still remains difficult for the following reasons:
}
\begin{enumerate}[(a)]

  \item  The domain $\imgof{\foldl(\seq, \zero)}$ can be infinite and in general not computable.\label{issue:uncomputable-domain}

  \item  Even if $\alpha$ and $\imgof{\foldl(\seq, \zero)}$ are computable, $\seq$ and $\comboper$ may not be computable. Na\"{i}vely enumerating elements in $\alpha$ and $\imgof{\foldl(\seq, \zero)}$ would not work.

  \item  Testing equality between elements of $\imgof{\foldl(\seq,\zero)\!}$ can be
    undecidable.

\end{enumerate}
Given $\seq :: \beta \times \alpha \to \beta$, recall that $\imgof{\foldl(\seq, \zero)}$ is a subset of $\beta$.
A sound but incomplete way to 
avoid
(\ref{issue:uncomputable-domain}) in practice is to
test the properties of $\comboper$ on all elements of $\beta$ instead.
If a~counterexample is found for some elements of $\beta$, the counterexample may not be valid in a~real \hscode{\kw{aggregate}} call because it may not belong to
$\imgof{\foldl(\seq, \zero)}$.
In practical cases, the sets $\alpha$ and $\beta$ are finite (such as machine integers) and equality between their elements is decidable. Even for such cases, checking if outcomes of $\agg$ are deterministic is still difficult since $\seq$ and $\comboper$ might not terminate for some input.
In many real Spark programs, however, $\seq$ and $\comboper$ are very simple and
thus computable (for instance, with only bounded loops or recursion).
A semi-procedure to test these conditions might work on such practical examples.



\vspace{-3.0mm}
\section{Case Studies} \label{section:examples}
\vspace{-2.0mm}
We evaluated advantages of our {\specspark} specification on several case studies.
In this section, we
first analyze a~Spark implementation of linear classification.
Using the \hscode{\kw{treeAggregate}} specification and its criteria for
deterministic outcomes, we construct inputs yielding non-deterministic outcomes
from the Spark implementation.
Second, we analyze an implementation of a standard scaler and find a
non-deterministic behavior there, too.
Yet more case studies are provided in App.~\ref{app:case_studies}.

\vspace{-3.0mm}
\subsection{Linear Classification}
\vspace{-2.0mm}

Linear classification is a well-known machine learning technique to
classify data sets. Fix a set of \emph{features}. A \emph{data
 point} is a vector of numerical feature values. A \emph{labeled} data
point is a data point with a discrete label. Given a labeled data set,
the \emph{classification problem} is to classify (new) unlabeled data
points by the labeled data set. A particularly useful subproblem is the
\emph{binary} classification problem. Consider, for instance, a 
data set of vital signs of some population; each data point is labeled
by the diagnosis of a disease (positive or negative). The binary
classification problem can be used to predict whether a person has the
particular disease. Linear classification solves the binary
classification problem by finding an optimal hyperplane to divide the
labeled data points. After a hyperplane is obtained,
linear classification predicts an unlabeled data point by the
half-space containing the point. Logistic regression and linear
Support Vector Machines (SVMs) are linear classification
algorithms. 

\enlargethispage{4mm}

Consider a~data set $\{ (\vec{x}_i, y_i) : 1 \leq i \leq n
\}$ of data points $\vec{x}_i \in \mathbb{R}^d$ labeled by
$y_i \in \{ 0, 1 \}$.
Linear classification can be expressed as a numerical optimization problem:
\vspace{-2mm}
\begin{equation*}
  \min\limits_{\vec{w} \in \mathbb{R}^d} f (\vec{w})
  \hspace{1em}
  \textmd{with}
  \hspace{1em}
  f (\vec{w}) =
  \xi R (\vec{w}) + 
  \frac{1}{n} \sum^n_{i=1} L (\vec{w}; \vec{x}_i, y_i)
\vspace{-2mm}
\end{equation*}
where $\xi \geq 0$ is a \emph{regularization parameter}, $R
(\vec{w})$ is a \emph{regularizer}, and $L (\vec{w}; \vec{x}_i, y_i)$
is a \emph{loss function}. A vector $\vec{w}$ corresponds to a
hyperplane in the data point space. The vector $\vec{w}_{\mathit{opt}}$ attaining the
optimum hence classifies unlabeled data points with criteria
defined by the objective function $f (\vec{w})$. Logistic regression and
linear SVM are but two instances of the optimization problem with
objective functions defined by different regularizers 
and loss functions.

In the Spark machine learning library, the numerical optimization
problem is solved by gradient descent. Very roughly, gradient descent
finds a local minimum of $f (\vec{w})$ by ``walking'' in the opposite
direction of the gradient of $f (\vec{w})$. 
The mean of subgradients
at data points is needed to compute the gradient of $f (\vec{w})$. The
Spark machine learning library invokes \hscode{\kw{treeAggregate}} to
compute the mean. Floating-point addition is used as
the \hscode{comb} parameter of the aggregate combinator. Since 
floating-point addition is not associative, we expect
to observe non-deterministic outcomes
(Proposition~\ref{lemma:treeAgg}). 

Consider the following three labeled data points: $-10^{20}$ labeled with $1$, $600$ labeled
with $0$, and $10^{20}$ labeled with $1$. We create a 20-partition RDD with an
equal number of the three labeled data points. The Spark machine
learning library function \hscode{LogisticRegressionWithSGD.train} is
used to generate a logistic regression model to predict the data
points $-10^{20}$, $600$, and $10^{20}$ in each run. Among 49 runs, 19
of them classify the three data points into two different classes: the
two positive data points are always classified in the same class, while
the negative data point in the other. The other 30 runs, however,
classify all three data points into the same class. We observe similar
predictions from \hscode{SVMWithSGD.train} with the same labeled data
points. 37 out of 46 runs classify the data points into two different
classes; the other 9 runs classify them into one class.
Interestingly, the data points are always classified
into two different classes by both logistic regression and linear SVM
when the input RDD has only three partitions. 
As we expected from our analysis of the function,
non-deterministic outcomes were witnessed in our Spark
distributed environment.

\vspace{-3.0mm}
\subsection{Standard Scaler}
\vspace{-2.0mm}
Standardization of data sets is a common pre-processing step in 
machine learning. 
Many machine learning algorithms
tend to perform better when the training set is similar to the standard
normal distribution.
In the Spark machine learning library, the class
\hscode{StandardScaler} is provided to standardize data sets.
The function \hscode{StandardScaler.fit} takes an RDD of raw data and
returns an instance of \hscode{StandardScalerModel} to transform data
points. Two transformations are available in \hscode{StandardScalerModel}. 
One standardizes a~data point by mean, and the other normalizes by 
variance of raw data. If data points in raw data are transformed by
mean, the transformed data points have the mean equal to~$0$.
Similarly, if they are transformed by variance, the transformed
data points have the variance~$1$.

The \hscode{StandardScaler} implementation uses \hscode{\kw{treeAggregate}} to
compute statistical information.
It uses floating-point addition to combine means of
raw data in different partitions.
As in the previous use case, since floating-point addition is not
associative, \hscode{StandardScaler} does not produce deterministic
outcomes (Section~\ref{sec:treeAggregate}). 
In our experiment, we create a 100-partition RDD with values
$-10^{20}, 600, 10^{20}$ of the same number of occurrences. The
mean of the data set is $(-10^{20} \times n + 600 \times n + 10^{20}
\times n) / (3n) = 200$ where $n$ is the number of occurrences of each
value. The data point $200$ should therefore be after standardization transformed to
$0$. In $50$ runs on the same data set in our distributed Spark platform,
\hscode{StandardScaler} transforms $200$ to a~range of values from $-944$
to $1142$, validating our prediction of a~non-deterministic outcome.

\enlargethispage{4mm}

\enlargethispage{4mm}

\vspace{-3.0mm}
\section{Related Work}
\vspace{-2.0mm}


{\bf MapReduce modeling and optimization.}
In the MapReduce (MR) computation,
various cost and performance models have been proposed~\cite{SLF13,herodotou2011profiling,modelMR,zhang13}.
These models estimate the execution time and resource requirements of MR jobs.
Karloff et al. developed a formal computation model for MR~\cite{ksv10} and showed how a variety of algorithms can exploit the combination of sequential and parallel computation in MR.
We are not aware of a~similar work in the context of Spark. To the best of our knowledge, our work is the first to address the problem of formal, functional specification of Spark aggregation.
Verifying the correctness of a MR program involves checking the commutativity and associativity of the reduce function.
Xu et al. propose various semantic criteria to model commonly held assumptions on MR programs~\cite{xu13semantic}, including determinism, partition isolation, commutativity, and associativity of map/reduce combinators.
Their empirical survey shows that these criteria are often overlooked by programmers and violated in practice.
A recent survey~\cite{xiao14mr} has found that a large number of industrial MR programs are, in fact, non-commutative.
Recent work has proposed techniques for checking commutativity of bounded reducers automatically~\cite{CHSW15}.
Because it is non-trivial to implement high-level algorithms using the MR framework, various approaches to compute optimized MR implementations have been proposed~\cite{GTA12,liu14,RFRS14}.
Emoto et al.~\cite{GTA12} formalize the algebraic conditions using semiring homomorphism, under which an efficient program based on the generate-test-aggregate programming model can be specified in the MR framework.
Given a monolithic {\it reduce} function, the work in~\cite{liu14} tries to decompose {\it reduce} into partial aggregation functions (similar to {\it seq} and {\it comb} in this paper) using program inversion techniques.
{\sc Mold}~\cite{RFRS14} translates imperative Java code into MR code by transforming imperative loops into {\it fold} combinators using semantic-preserving program rewrite rules.

\noindent
{\bf Numerical Stability under MapReduce.} Several works try to scale up machine learning algorithms for large datasets using MapReduce~\cite{mrml06,SLF13}.
To achieve numerically stable results across multiple runs~\cite{bennett09,tian12}, for example, preventing overflow, underflow and round-off errors due to finite-precision arithmetic, a variety of techniques are proposed~\cite{tian12}: generalizing sequential numerical stability techniques to distributed settings, shifting data values by constants, divide-and-conquer, etc.
We showed that simulating machine learning algorithms using our specification enables early detection of points of numerical instability.

\noindent
{\bf Relational Query Optimization.}
Relational query optimization is an extensively researched topic~\cite{Chaudhuri98,Io96}: the goal is to obtain equivalent but more efficient query expressions by exploiting the algebraic properties of the constituent operators, for instance, join, select, together with statistics on relations and indices. 
For example, while inner joins commute independent of data, left joins commute only in specific cases.
Query optimization for partitioned tables has received less attention~\cite{Hero11,db2partitioned}: because the key relational operators are not partition-aware, most work has focused on necessary but not sufficient conditions for query equivalence.
In contrast, we investigate determinism of Spark aggregate expressions, constructed using partition-aware {\it seq} and {\it comb} combinators.
We describe necessary and sufficient conditions under which these computations yield deterministic results independent of the data partitions.

\noindent
{\bf Deterministic Parallel Programming.}
In order to enable deterministic-by-default parallel programming~\cite{bocchino09,BS10,bocchino11,budimlic10,leijen11}, researchers have developed several programming abstractions and logical specification languages to ensure that programs produce the same output for the same input independent of thread scheduling.
For example, Deterministic Parallel Java~\cite{bocchino09,bocchino11} ensures exclusive writes to shared memory regions by means of verified, user-provided annotations over memory regions.
In contrast, deterministic outcomes from Spark aggregation depend on algebraic properties like commutativity and associativity of {\it seq} and {\it comb} functions and their interplay

\comment{

R. L. Bocchino, Jr., V. S. Adve, D. Dig, S. V. Adve, S. Heumann,
R. Komuravelli, J. Overbey, P. Simmons, H. Sung, and M. Vakilian.
A type A type and effect system for deterministic parallel Java. In OOPSLA, 2009

Z. Budimlic, M. Burke, V. Cav ´ e, K. Knobe, G. Lowney, R. Newton, ´
J. Palsberg, D. Peixotto, V. Sarkar, F. Schlimbach, and S. Tas¸irlar.
Concurrent Collections. Sci. Program., 18, August 2010.

D. Leijen, M. Fahndrich, and S. Burckhardt. Prettier concurrency:
purely functional concurrent revisions. In Haskell, 2011.

Map Reduce:

X Sakr S, Liu A, Fayoumi AG. The family of MapReduce and large-scale data processing systems. ACM Comp. Surveys Oct 2013; 46(1):44 pp. Article 11. 5, 36

Map Reduce + Machine Learning:
C. Chu, S. Kim, Y. Lin, Y. Yu, G. Bradski, A. Ng, and K. Olukotun.
Map-reduce for machine learning on multicore. NIPS 2006

D. Gillick, A. Faria, and J. DeNero. Mapreduce: Distributed computing
for machine learning, 2008.

--
Scalable and Numerically Stable Descriptive
Statistics in SystemML 
Yuanyuan Tian Shirish Tatikonda Berthold Reinwald
ICDE'12

J. Bennett, R. Grout, P. P´ebay, D. Roe, and D. Thompson. Numerically
stable, single-pass, parallel statistics algorithms. In IEEE International
Conference on Cluster Computing and Workshops, pages 1–8. IEEE,
2009.

---- automated verification of distributed system properties
Verifying Eventual Consistency of
Optimistic Replication Systems
Ahmed Bouajjani Constantin Enea Jad Hamza

tla+ at amazon.

formalization of TSO and other weak memory execution architectures.

Given a set M and an associative binary operator $\bop$ on M with
its identity $\idelm$, the pair $(M,\bop)$ is called a {\it monoid}. For example,
the set of integers with the usual plus operator forms a monoid
$(Z,+)$. 
}

\vspace{-3.0mm}
\section{Conclusion} \label{section:conclusion}
\vspace{-2.0mm}

In this paper, we give a~Haskell specification for various Spark
aggregate combinators.
We focus on aggregation of RDDs representing general sets, sets of pairs, and
graphs.
Based on our specification, we derive necessary and sufficient conditions that guarantee deterministic
outcomes of the considered Spark aggregate combinators.
We investigate several case studies and use the conditions to predict
non-deterministic outcomes. 
Our executable specification can be used by developers
for more detailed analysis and efficient development of
distributed Spark programs.
We also believe that our specifications are valuable resources for research communities to understand Spark better.

\hide{
We present a Haskell specification for various Spark aggregate
combinators for values, pairs, and graph RDDs. Based on our
specification, requirements for deterministic outcomes from Spark
aggregation are characterized. Non-deterministic outcomes predicted by
our characterizations are witnessed. Our executable specification
moreover enables more detailed analysis and efficient development of
distributed Spark programs.
}

There are several future directions.
The conditions for deterministic outcomes of aggregate combinators could be
used for:
\begin{inparaenum}[(i)]
  \item  creating fully mechanized proofs for properties about data-parallel programs;
  \item  developing automatic techniques for detecting non-deterministic outcomes of 
    data-parallel programs; and
  \item  synthesizing deterministic concurrent programs
    from sequential specifications.
\end{inparaenum}
We have formalized the
proofs of some crucial lemmas in Agda~\cite{purespark}.
Using Scalaz~\cite{scalaz-github}, verified Haskell specifications 
can be translated to Spark programs to ensure determinism by construction.

\hide{
Lots of opportunities follow from our first executable formal
specification. Our specification could be adopted to carry out
mechanized proofs about Spark aggregation with expressive type
systems. If Spark programs adopted our specification, their properties
could also be proved formally. Using Scalaz~\cite{scalaz-github}, Haskell specifications 
moreover could be translated to Spark programs. One could analyze
Spark programs formally.

Our specification also opens the door for automatic verification. 
Non-deterministic outcomes are predicted and observed by our
characterization. Such a deviant behavior is unique and undesirable in
data parallel computation. Automatic program analysis techniques could
be developed to identify such unexpected outcomes in Spark programs. 
Other undesirable behaviors unique in data parallel computation could
also be identified from our specification.
}

\enlargethispage{4mm}

\smallskip\noindent\emph{Acknowledgement.}
This work was supported by
the Czech Science Foundation (project 17-12465S),
the BUT FIT project FIT-S-17-4014,
the IT4IXS: IT4Innovations Excellence in Science project (LQ1602), and
Ministry of Science and Technology, R.O.C. (MOST projects 103-2221-E-001-019-MY3 and 103-2221-E-001-020-MY3).

\enlargethispage{2mm}

\vspace{-3mm}
\bibliographystyle{splncs03}
\bibliography{literature}

\clearpage
\appendix

\vspace{-0.0mm}
\section{Graph RDDs}\label{app:graph_rdds}
\vspace{-0.0mm}

Using RDDs, Spark provides a framework to analyze graphs distributively.
In the Spark GraphX library, each vertex in a graph is designated by a
\hscode{VertexId}, and associated with a~vertex attribute.
Each edge on the other hand is represented by \hscode{VertexId}s of its source and destination
vertices. An edge is also associated with an edge attribute.
\vspace{-1mm}
\begin{lstlisting}[language=Haskell,basicstyle=\sffamily\footnotesize]
type VertexId = Int
type VertexRDD $\alpha$ = PairRDD VertexId $\alpha$
type EdgeRDD $\beta$ = RDD (VertexId, VertexId, $\beta$)
data GraphRDD $\alpha$ $\beta$ = Graph { vertexRdd :: VertexRDD $\alpha$, edgeRdd :: EdgeRDD $\beta$ }
\end{lstlisting}
\vspace{-1mm}
Let \hscode{graphRdd} be a graph RDD. Its vertex RDD \hscode{(vertexRdd
 graphRdd)} contains pairs of vertex identifiers and
attributes. Different from conventional pair RDDs, each vertex
identifier can appear at most once in the vertex RDD since a vertex is
associated with exactly one attribute. If, for
instance, two pairs with the same vertex identifier are generated during
computation, their associated attributes must be merged
to obtain a~valid vertex RDD. The edge RDD \hscode{(edgeRdd graphRdd)} 
consists of triples of source and destination vertex identifiers, and 
edge attributes. Multi-edged directed graphs are allowed.
In a graph RDD, the vertex and edge RDDs need to be
consistent. That is, the source and destination vertex identifiers of
any edge from the edge RDD must appear in the vertex RDD of the graph RDD.

\enlargethispage{2mm}

The Spark GraphX library provides aggregate combinators for
graph RDDs. We begin with an informal description of a slightly
more general 
\hscode{\kw{aggregateMessagesWithActiveSet}}
combinator~(Algorithm~\ref{algorithm:pseudo-code-aggregateMessages}).
The combinator takes functions
\hscode{sendMsg} and \hscode{mergeMsg}, and a list
\hscode{active} of vertices as its parameters. 
The list \hscode{active} determines \emph{active} edges, that is,
edges with source or destination vertex identifiers in
\hscode{active}. For each active edge, the function
\hscode{\kw{aggregateMessagesWithActiveSet}} invokes \hscode{sendMsg} to
send messages to its vertices. Messages sent to each vertex are merged
by \hscode{mergeMsg}. Since a vertex is associated with at most one
message after merging, the result is a~valid vertex RDD.

\begin{algorithm}[H]
\footnotesize
{

  \ForEach{active edge $e$}
  {
    call \textsf{sendMsg} on $e$ to send messages to vertices of $e$\;
  }
  \ForEach{vertex $v$ receiving messages}
  {
    call \textsf{mergeMsg} to merge all messages sent to $v$\;
  }
  \Return a vertex RDD with merged messages\;
}
  \vspace{1mm}
  \label{algorithm:pseudo-code-aggregateMessages}
  \caption{\hscode{\kw{aggregateMessagesWithActiveSet}}}
\end{algorithm}

Formally, the function \hscode{sendMsg} accepts source and destination 
vertex identifiers, attributes of the vertices, and the edge attribute of an
edge as inputs. It sends messages to the source or destination vertex, 
both, or none.
In our specification,
\hscode{\kw{lookup}} is used to obtain vertex attributes from a vertex
RDD. We generate a pair RDD of vertex identifiers and messages by invoking \hscode{sendMsg} on every
active edge. The messages associated with the same vertex are
then merged by applying \hscode{\kw{reduceByKey}} on the pair
RDD. The resultant vertex RDD contains merged messages as vertex
attributes. We call it a \emph{message} RDD for clarity.
Note that if a vertex from the input graph RDD does not receive
any message, it is not present in the output message RDD. The 
combinator \hscode{\kw{aggregateMessages}} in the Spark GraphX
library is defined by \hscode{\kw{aggregateMessagesWithActiveSet}}. 
It invokes \hscode{\kw{aggregateMessagesWithActiveSet}} by passing the list of
all vertex identifiers as the \hscode{active} list.
The combinator effectively applies \hscode{sendMsg} to every 
edge in a graph RDD. 
\begin{lstlisting}[language=Haskell,basicstyle=\sffamily\footnotesize,morekeywords={aggregateMessagesWithActiveSet,reduceByKey,aggregateMessages}]
aggregateMessagesWithActiveSet ::
         (VertexId $\rightarrow$ $\alpha$ $\rightarrow$ VertexId $\rightarrow$ $\alpha$ $\rightarrow$ $\beta$ $\rightarrow$ [(VertexId, $\gamma$)]) 
         $\rightarrow$ ($\gamma$ $\rightarrow$ $\gamma$ $\rightarrow$ $\gamma$) $\rightarrow$ [VertexId] $\rightarrow$ GraphRDD $\alpha$ $\beta$ $\rightarrow$ VertexRDD $\gamma$
aggregateMessagesWithActiveSet sendMsg mergeMsg active graphRdd =
  let isActive (srcId, dstId, _) = srcId `elem` active || dstId `elem` active
     vAttrs = concat (vertexRdd graphRdd)
     f edge = if isActive edge then
                let (srcId, dstId, edgeAttr) = edge
                   srcAttr = fromJust (lookup srcId vAttrs)
                   dstAttr = fromJust (lookup dstId vAttrs)
                in sendMsg srcId srcAttr dstId dstAttr edgeAttr
              else []
     pairRdd = map (concatMap f) (edgeRdd graphRdd)
  in reduceByKey mergeMsg pairRdd

aggregateMessages :: (VertexId $\rightarrow$ $\alpha$ $\rightarrow$ VertexId $\rightarrow$ $\alpha$ $\rightarrow$ $\beta$ $\rightarrow$ [(VertexId, $\gamma$)])
         $\rightarrow$ ($\gamma$ $\rightarrow$ $\gamma$ $\rightarrow$ $\gamma$) $\rightarrow$ GraphRDD $\alpha$ $\beta$ $\rightarrow$ VertexRDD $\gamma$
aggregateMessages sendMsg mergeMsg graphRdd =
  let vertices = concatMap (map fst) (vertexRdd graphRdd)
  in aggregateMessagesWithActiveSet sendMsg mergeMsg vertices graphRdd 
\end{lstlisting}


Many graph algorithms perform fixed point computation. 
The Spark GraphX library hence provides a Pregel-like function to
apply \hscode{\kw{aggregateMessages}} on a graph RDD
repetitively~\cite{MABDHLC:2010:PLSGP}. The Spark
\hscode{\kw{pregel}} function takes four input parameters \hscode{initMsg}, \hscode{vprog},
\hscode{sendMsg}, and \hscode{mergeMsg}~(Algorithm~\ref{algorithm:pseudo-code-pregel}). 
At initialization, 
it updates vertex attributes of the graph RDD by invoking
\hscode{vprog} with the initial message \hscode{initMsg}. The
\hscode{\kw{pregel}} function then calls \hscode{\kw{aggregateMessages}} to
obtain a message RDD.
If a vertex 
receives a message, its attribute is updated by \hscode{vprog} with
the message. After updating vertex attributes, \hscode{\kw{pregel}} obtains
a new message RDD by invoking
\hscode{\kw{aggregateMessagesWithActiveSet}} with the active list equal to
message-receiving vertices. Subsequently, 
only edges connecting to such vertices can send new messages.  

\DecMargin{2mm}
\begin{algorithm}[H]
{\footnotesize
  \ForEach{vertex $v$ in G}
  {
    call \textsf{vprog} on $v$ with \textsf{initMsg} to obtain its
    initial vertex attribute\; 
  }
  msgRdd $\leftarrow$ call \hscode{\kw{aggregateMessages}} on G\; 
  \While{msgRdd is not empty}
  {
    \ForEach{vertex $v$ with message $m$ in msgRdd}
    {
      call \textsf{vprog} on $v$ with $m$ to update its vertex attribute on G\;
    }
    msgRdd $\leftarrow$ call \hscode{\kw{aggregateMessagesWithActiveSet}} 
    with \textsf{active} equal to the vertices in msgRdd\;
  }
  \Return G\;
}
  \vspace{-1mm}
  \caption{\hscode{\kw{pregel}}}
  \label{algorithm:pseudo-code-pregel}
\end{algorithm}


We use several auxiliary functions to specify the Spark
\hscode{\kw{pregel}} function. Given a function computing an attribute from
a vertex identifier and an attribute, the auxiliary
function \hscode{\kw{mapVertexRDD}} applies the function to every vertex in
a vertex RDD and obtains another vertex RDD with new attributes.
The \hscode{\kw{mapVertexRDD}} function is used in \hscode{\kw{mapVertices}} to update vertex attributes in
graph RDDs. Moreover, recall that \hscode{\kw{aggregateMessagesWithActiveSet}}
returns a message RDD. The auxiliary function
\hscode{\kw{joinGraph}} updates a graph RDD with messages in a message
RDD. For each vertex in the graph RDD, its attribute is joined with
the message in the message RDD. If there is no message, the
vertex attribute is left unchanged. The \hscode{\kw{pregel}} function
sets up the initial graph RDD by \hscode{\kw{mapVertices}}. It then
computes the initial message RDD by
\hscode{\kw{aggregateMessages}}. In each iteration, a new graph RDD is
obtained by joining the graph RDD with a message
RDD. \hscode{\kw{aggregateMessagesWithActiveSet}} is then invoked to compute 
a new message RDD for the next iteration. The \hscode{\kw{pregel}} function
terminates when no more message is sent.
\begin{lstlisting}[language=Haskell,basicstyle=\sffamily\footnotesize, columns=fullflexible,morekeywords={mapVertexRDD,mapVertices,joinGraph,pregel,aggregateMessages,aggregateMessagesWithActiveSet}]
mapVertexRDD :: (VertexId $\rightarrow$ $\alpha$ $\rightarrow$ $\beta$) $\rightarrow$ VertexRDD $\alpha$ $\rightarrow$ VertexRDD $\beta$
mapVertexRDD f vRdd = map (map ($\lambda$(i, attr) $\rightarrow$ (i, f i attr))) vRdd

mapVertices :: (VertexId $\rightarrow$ $\alpha$ $\rightarrow$ $\gamma$) $\rightarrow$ GraphRDD $\alpha$ $\beta$ $\rightarrow$ GraphRDD $\gamma$ $\beta$
mapVertices updater gRdd = Graph {
  vertexRdd = mapVertexRDD updater (vertexRdd gRdd),
  edgeRdd = edgeRdd gRdd }

joinGraph :: (VertexId $\rightarrow$ $\alpha$ $\rightarrow$ $\gamma$ $\rightarrow$ $\alpha$) $\rightarrow$ GraphRDD $\alpha$ $\beta$
            $\rightarrow$ VertexRDD $\gamma$ $\rightarrow$ GraphRDD $\alpha$ $\beta$
joinGraph joiner gRdd msgRdd = let assoc = concat msgRdd
     updt i attr = case lookup i assoc of Just v  $\rightarrow$ joiner i attr v
                                       $\,$Nothing $\rightarrow$ attr
  in mapVertices updt gRdd

pregel :: $\gamma$ $\rightarrow$ (VertexId $\rightarrow$ $\alpha$ $\rightarrow$ $\gamma$ $\rightarrow$ $\alpha$) $\rightarrow$
        (VertexId $\rightarrow$ $\alpha$ $\rightarrow$ VertexId $\rightarrow$ $\alpha$ $\rightarrow$ $\beta$ $\rightarrow$ [(VertexId, $\gamma$)])
        $\rightarrow$ ($\gamma$ $\rightarrow$ $\gamma$ $\rightarrow$ $\gamma$) $\rightarrow$ GraphRDD $\alpha$ $\beta$ $\rightarrow$ GraphRDD $\alpha$ $\beta$
pregel initMsg vprog sendMsg mergeMsg graphRdd =
  let initG = let init_f i attr = vprog i attr initMsg 
             in mapVertices init_f graphRdd
     initMsgRdd = aggregateMessages sendMsg mergeMsg initG
     loop curG [] = curG
     loop curG msgRdd = let newG = joinGraph vprog curG msgRdd
          active = concatMap (map fst) msgRdd
          msgRdd' = aggregateMessagesWithActiveSet
                       sendMsg mergeMsg active newG
       in loop newG msgRdd'
  in loop initG initMsgRdd
\end{lstlisting}

\vspace{-0.0mm}
\subsection{Deterministic Aggregation in Graph Rdds}\label{sec:label}
\vspace{-0.0mm}

In this section, we explore necessary and sufficient conditions for aggregation
in graph RDDs.
%
%
%
In particular, we investigate deterministic outcomes of calls to the function
$\aggMsgs(\send, \comboper, \graphrdd)$ for $\send :: \vertexid \times \alpha \times \vertexid \times \alpha \times \beta \to \listof{(\vertexid, \gamma)}$, 
$\comboper :: \gamma \times \gamma \to \gamma$, 
and $\graphrdd :: \GraphRDD\ \alpha\ \beta$.
We define deterministic outcomes first.
\begin{definition}
Calls to the function $\aggMsgs(\!\send\!,\!\comboper,\!\graphrdd\!)$ have
\emph{deterministic outcomes} if for any
two graph RDD representations of the same graph
\vspace{-1ex}
\begin{align*}
\MoveEqLeft
\graphrdd_1, \graphrdd_2 :: \hscode{GraphRDD}\ \alpha\ \beta,
\end{align*}
we have for all vertex identifiers $v :: \vertexid$,
\vspace{-1ex}
\begin{align*}
\lookup(v, \aggMsgs(\send, \comboper, \graphrdd_1))=\\
\lookup(v, \aggMsgs(\send, \comboper, \graphrdd_2)).
\end{align*}
\end{definition}

The following proposition gives a sufficient condition for $\aggMsgs$ to
have deterministic outcomes.
\begin{restatable}{proposition}{lemDetAggMsg}\label{lem:lemDetAggMsg}
It holds that if calls to the function $\redBKey(\comboper, \rdd)$ have deterministic outcomes,
then calls to the function $\aggMsgs(\send,
\comboper, \graphrdd)$ also have deterministic outcomes.
\end{restatable}

\vspace{-0.0mm}
\section{Extended Set of Case Studies}\label{app:case_studies}
\vspace{-0.0mm}

This section of the appendix gives yet more case studies that we explored when
analyzing Spark's machine learning and graph libraries.

\vspace{-0.0mm}
\subsection{Vertex Coloring}
\vspace{-0.0mm}
Let $\Gamma = \{1,...,k\}$ denote the set of $k$ \emph{colors}.
Given an undirected graph $G=(V,E)$, a \emph{$k$-coloring} of $G$ is a
map $C: V \rightarrow \Gamma$ such that $C(v) \neq C(u)$ for any
$\{ v, u \} \in E$. 
In this case study, we will implement the Communication-Free Learning
(CFL) algorithm~\cite{leith2006convergence} to find a $k$-coloring using
the Spark GraphX library. Let $0 < \beta < 1$. 
The algorithm computes a $k$-coloring by iterations. 
We say a vertex $v$ is \emph{inactive} if all vertices
adjacent to $v$ 
have colors different from the color of $v$.
Otherwise, $v$ is \emph{active}. 
At the $n$-th iteration, the CFL
algorithm randomly chooses a color $C_n(v) \in \Gamma$ by the color
distribution $P_n (v, \bullet)$ of $v$. 
The color distribution $P_n (v, \bullet)$ is defined as follows. 
For $n = 0$, $P_0(v, c) = 1/k$ for
all $v \in V$ and $c \in \Gamma$. Each vertex hence chooses one
of the $k$ colors uniformly at random. For $n > 0$, let $c = C_{n-1}
(v)$ be the color of $v$ in the previous iteration.
\begin{itemize}
\item If $v$ is inactive, define $P_n(v, c) = 1$ and $P_n(v, d) = 0$
  for $d \neq c$. Thus $v$ does not change its color.
\item Otherwise, define
\vspace{-1ex}
  \begin{equation*}
    \hspace{-3.2mm}P_{n}(v, d) =
    \begin{cases}
      (1-\beta)\cdot P_{n-1} (v, c) & \textmd{if } d = c\\
      (1-\beta)\cdot P_{n-1}(v,d)+\beta/(k-1)\!\!\! & \textmd{if } d \neq c
    \end{cases}
  \end{equation*}
  Thus $v$ is more likely to choose a color different from~$c$.
\end{itemize}
Observe that $C_n$ stabilizes if and only if it is a $k$-coloring. 

\hide{
The CFL algorithm can be parallelized as follows. 
\begin{enumerate}
\item Initialize $C_0(v)$ and $P_0(v,\bullet)$ for each vertex $v$.
\item A) An active vertex $v$ is deactivated if no vertices in $N(v)$ share the same color with $v$. 
        An inactive vertex $v$ is activated if there is a vertex in $N(v)$ sharing the same color with $v$.
        \vspace{.5em} \\
        B) Each vertex $v$ computes $C_n(v)$ and $P_n(v,\bullet)$ upon the $n$-th change of the vertex state (from active to inactive or vice versa).
  \item Repeat Step 2 until all vertices are simultaneously inactive.
\end{enumerate}
}

We implement the CFL algorithm using \hscode{\kw{pregel}} in {\specspark}.
For each vertex $v$, its attribute consists
of the vertex color $C_n(v)$, the color distribution $P_n(v,\bullet)$,
the vertex state (active or not), and a random number generator.
As in
Section~\ref{section:triangle-count}, 
an edge \hscode{($u$, $v$, $\uscore$)} with $u \geq v$ in an edge RDD
represents $\{ u, v \} \in E$.
Given a graph RDD \hscode{graphRdd}, we construct its base graph
\hscode{baseG} with
initial vertex attributes.
\vspace{-1mm}
\begin{lstlisting}[language=Haskell,basicstyle=\sffamily\footnotesize]
initDist = map ($\lambda$_ $\rightarrow$ 1.0 / fromIntegral k) [1..k]

baseG = mapVertices ($\lambda$i _ $\rightarrow$ let (c, g) = randomR (1, k) (mkStdGen i)
                             in (c, initDist, True, g)) graphRdd
\end{lstlisting}
\vspace{-1mm}
where \hscode{initDist} is the uniform distribution over $k$ colors.

Consider the following \hscode{sendMsg} function:
\vspace{-1mm}
\begin{lstlisting}[language=Haskell,basicstyle=\sffamily\footnotesize]
sendMsg srcId (srcColor, _, srcActive, _) dstId (dstColor, _, dstActive, _) _ =
  if srcColor == dstColor then [(srcId, True), (dstId, True)]
  else (if srcActive then [(srcId, False)] else []) $\cat$
      (if dstActive then [(dstId, False)] else [])
mergeMsg msg1 msg2 = msg1 || msg2
\end{lstlisting}
\vspace{-1mm}

If the source and destination vertices of an edge have the same color,
\hscode{sendMsg} sends \hscode{\kw{True}} to both vertices to update
vertex attributes. 
If they have different colors and the source vertex is
active, \hscode{\kw{False}} is sent to the source vertex. Similarly,
\hscode{\kw{False}} is sent to the destination vertex if the vertex is
active. \hscode{mergeMsg} is the disjunction of
messages. After applying \hscode{\kw{aggregateMessagesWithActiveSet}} with
\hscode{sendMsg} and \hscode{mergeMsg}, a vertex may
receive a Boolean message. If a vertex receives \hscode{\kw{True}}, it
becomes active since one of its neighbors has the same color. 
Otherwise, the vertex becomes inactive.

We use \hscode{vprog} to update vertex attributes. For each vertex
receiving a message, its vertex state, color, and color distribution
are updated according to the CFL algorithm.
The auxiliary function \hscode{sampleColor}
chooses a color randomly by the color distribution. 
The \hscode{helper} function in \hscode{vprog} computes the 
color distribution $P_n (v, \bullet)$ for the next iteration.
\vspace{-1mm}
\begin{lstlisting}[language=Haskell,basicstyle=\sffamily\footnotesize]
sampleColor dist p = let f (color, mass) weight =
       $\ $(if m < p then succ color else color, m)
           where m = mass + weight
  in fst (foldl f (1, 0.0) dist)

vprog _ (c, dist, _, g) active = let helper (i, res) weight =
        let decay = weight * (1 - beta)
           d = decay + (if c == i then 0 else beta / fromIntegral (numColors-1))
           e = if c == i then 1.0 else 0.0
        in (succ i, if active then res $\cat$ [d] else res $\cat$ [e])
     dist' = snd (foldl helper (1, []) dist)
     (p, g') = random g
     c' = if active then sampleColor dist' p else c
  in (c', dist', active, g')
\end{lstlisting}
\vspace{-1mm}

Finally, we invoke \hscode{\kw{pregel}} to compute a $k$-coloring:
\vspace{-1mm}
\begin{lstlisting}[language=Haskell,basicstyle=\sffamily\footnotesize,morekeywords={pregel}]
coloring = pregel True vprog sendMsg mergeMsg baseG
\end{lstlisting}
\vspace{-1mm}
We test our executable Haskell specification on a typical Linux server. 
Since our Spark specification {\specspark}  is faithful to Spark APIs, 
we realize it in the GraphX library with little manual effort. 
Our implementation works as intended on the distributed Spark platform.

\hide{
It is shown in \cite{duffy2008complexity} that the CFL algorithm colors a graph in exponential time with high probability. More precisely, given any $\epsilon \in (0,1)$, the algorithm can find a $k$-coloring for a $k$-colorable graph with probability $1-\epsilon$ in $O(N\exp^{cN\lg 1/\epsilon})$ iterations, where $N$ is the number of vertices and $c$ is a constant depending on $\beta$.
}

\subsection{Connected Components}
The Spark GraphX library implements a connected component algorithm
for direct graphs. The documentation however does not explain what
connected components are in directed graphs. We will find out
what the implementation does here. Consider
the following {\specspark} specification extracted from the
Spark implementation: 
\vspace{-1mm}
\begin{lstlisting}[language=Haskell,basicstyle=\sffamily\footnotesize,morekeywords={pregel,mapVertices,connectedComponent}]
connectedComponent graphRdd = 
  let baseG = mapVertices ($\lambda$i _ $\rightarrow$ i) graphRdd
     initMsg = maxBound :: Int
     sendMsg src srcA dst dstA _ =
       if srcA < dstA then [(dst, srcA)]
       else if dstA < srcA then [(src, dstA)]
       else []
     vprog _ attr msg = min attr msg
  in pregel initMsg vprog sendMsg min baseG
\end{lstlisting}
\vspace{-1mm}
Given a graph RDD \hscode{graphRdd}, its base graph \hscode{baseG} is obtained by
setting the attribute of a vertex to the identifier of the vertex. 
\hscode{sendMsg} compares the attributes of the source and destination
vertices of an edge. The smaller attribute is sent to the vertex with the larger
attribute. If both attributes are equal, no message is sent. If a
number of messages are sent to a vertex, only the minimal message
remains after applying \hscode{\kw{aggregateMessagesWithActiveSet}} with
\hscode{sendMsg} and \hscode{\kw{min}}. When a vertex receives a
message, its attribute is set to the minimum of its attribute and
the message.

Consider a graph $G = (V, E)$ with $E \subseteq V \times V$. We 
use $\attr{v}$ for the attribute of the vertex $v \in V$. Two vertices
$u$ and $v$ are \emph{linked} if $(u, v) \in E$ or $(v, u) \in E$.
Using our specification of \hscode{\kw{pregel}}, it is not hard to see that the
{\specspark} specification implements Algorithm~\ref{algorithm:connected-components}.
Note that the two for-each loops essentially propagate
minimal attributes to linked vertices. When the set
$\mathit{active}$ is empty, the attributes of every linked vertices
are equal and the algorithm terminates. We say two vertices $u$ and $v$ are
\emph{connected} if there are $w_0 = u, w_1, \ldots, w_k = v$
such that $w_i$ and $w_{i+1}$ are linked for $0 \leq i < k$. When
\hscode{\kw{connectedComponent}} terminates, connected vertices
have the same attribute equal to the minimal vertex
identifier among them. Hence the Spark
implementation returns a graph RDD whose vertex attributes are the
minimal vertex identifiers of connected vertices.

\begin{algorithm}
{\footnotesize
  $\attr{v}$ $\leftarrow$ the vertix identifier of $v$\;
  $\mathit{active}$ $\leftarrow$ $V$\;
  \While{$\mathit{active} \neq \emptyset$}
  {
    $\mathit{active}'$ $\leftarrow$ $\emptyset$\;
    \ForEach{$v \in \mathit{active}$}
    {
      \uIf{$\attr{u} < \attr{v}$ for some $u$ linked with $v$}
      {
        send $\attr{u}$ to $v$ and add $v$ to $\mathit{active}'$
      }
      \uIf{$\attr{v} < \attr{u}$ for some $u$ linked with $v$}
      {
        send $\attr{v}$ to $u$ and add $u$ to $\mathit{active}'$
      }
    }
    \ForEach{$v \in \mathit{active}'$}
    {
      $\attr{v}$ $\leftarrow$ the minimal attribute sent to $v$
    }
    $\mathit{active}$ $\leftarrow$ $\mathit{active}'$\;
  }
}  
  \caption{\hscode{\kw{connectedComponents}}}
  \label{algorithm:connected-components}
\end{algorithm}

One can informally reason that the {\specspark} connected component specification has deterministic outcomes.
Note that \hscode{(VertexId, min)} is a commutative semigroup. This allows us to derive a similar proposition for \hscode{\kw{aggregateMessagesWithActiveSet}}. The calls to \hscode{\kw{aggregateMessages}} and \hscode{\kw{aggregateMessagesWithActiveSet}} in \hscode{\kw{pregel}} therefore have deterministic outcomes (Proposition~\ref{lem:lemDetAggMsg}). 
Examining the \hscode{vprog} in our connected component specification, the functions \hscode{\kw{mapVertices}} and \hscode{\kw{joinGraph}} also have deterministic outcomes. All potential sources of non-determinism in \hscode{\kw{pregel}} have deterministic outcomes. The connected component specification consequently has deterministic outcomes. Experiments in a distributed Spark environment confirm our reasoning.

\subsection{Triangle Count}
\label{section:triangle-count}
Let $G = (V, E)$ be an undirected graph without self-loops or multiple
edges. For $u, v \in V$, $\{ u, v \} \in E$ denotes that $u$ and $v$
are adjacent. A 
\emph{triangle} in $G$ is formed by $u, v, w \in V$ such that $\{ u, v \},
\{ u, w \}, \{ v, w \} \in E$. Counting the number of triangles 
is important to, for example, network analysis. The Spark GraphX library implements
the triangle counting algorithm using \hscode{\kw{aggregateMessages}}.

In the GraphX implementation, an undirected graph is represented by a
graph RDD where the source vertex identifier of every edge is greater than its
destination vertex identifier. An edge $\{ u, v
\} \in E$ with $u > v$ is thus represented by \hscode{($u$, $v$, \_)} in
an edge RDD. Below is the {\specspark}
specification extracted from the Spark GraphX implementation.
\vspace{-2mm}
\begin{lstlisting}[language=Haskell,basicstyle=\sffamily\footnotesize,morekeywords={aggregateMessages,mapVertices,intersection,mapVertexRDD}]
sendMsg src _ dst _ _ = [(dst, singleton src), (src, singleton dst)]
adjacentVRdd = aggregateMessages sendMsg (union) graphRdd

newGRdd = let adjacents = concat adjacentVRdd
              $\ \ $updt v _ = case lookup v adjacents of
                            Just adj $\rightarrow$ delete v adj
                            Nothing $\rightarrow$ empty
            in mapVertices updt graphRdd

sendMsg2 src srcA dst dstA _ =
  let num = size (intersection srcA dstA) 
  in [(dst, num),(src,num)]
sumTriangles = aggregateMessages sendMsg2 (+) newGRdd

triangleCount = mapVertexRDD ($\lambda$_ y $\rightarrow$ quot y 2) sumTriangles
\end{lstlisting}
\vspace{-1mm}

For each edge $\{ u, v \} \in E$, \hscode{sendMsg} sends $\{ u \}$
and $\{ v \}$ to vertices $v$ and $u$ respectively. Multiple messages
to a vertex are merged by union. After applying
\hscode{\kw{aggregateMessages}} with \hscode{sendMsg} and \hscode{\kw{union}},
\hscode{adjacentVRdd} is a vertex RDD where the attribute of the vertex $v$ 
is $\{ u : \{ u, v \} \in E \}$.

The implementation updates vertex attributes of the input graph
to obtain \hscode{newGRdd}. 
If the set $A$ of vertices
adjacent to $v$ is not empty, the attribute of $v$ is updated to $A
\setminus \{ v \}$. If $v$ does not have any adjacent vertices, its
attribute is set to the empty set. Hence the attribute of a vertex in
\hscode{newGRdd} contains its adjacent vertices but not itself.
Recall that we assume the input graph
does not have self-loops. A vertex cannot be
adjacent to itself.
Removing a vertex from the set of its adjacent vertices is redundant.

For each edge $\{ u, v \} \in E$ in \hscode{newGRdd}, \hscode{sendMsg2}
sends the message $| U \cap V |$ to $u$ and $v$ where $U$ and $V$ are
the sets of vertices adjacent to $u$ and $v$ respectively. Observe that
for every $w \in U \cap V$, we have $\{ w, u \}, \{ w, v \}, \{ u, v
\} \in E$. Let $\numTri{\{u, v\}}$ denote the number of triangles
containing the edge $\{ u, v \}$. $\numTri{\{u, v\}}$ is sent to
both $u$ and $v$. Messages are moreover merged
by summation. Hence the attribute of each vertex $v$ in
\hscode{sumTriangles} is $\sum_{\{ u, v \} \in E} \numTri{\{u, v\}}$.

Now consider a vertex $v$ in a triangle of $u, v, w$. The triangle is
counted in both $\numTri{\{ u, v \}}$ and $\numTri{\{ w, v \}}$. 
Since a triangle is always counted twice, 
the attribute  given as $\frac{1}{2} \sum_{\{ u, v \} \in E} \numTri{\{u, v\}}$
of vertex $v$ in \hscode{triangleCount} is the the number of triangles
containing $v$. 
Both calls to \hscode{\kw{aggregateMessages}} have deterministic outcomes because the algebras $(\hscode{Set}, \hscode{(union)})$ and $(\hscode{Int}, \hscode{(+)})$ are commutative semigroups (Propositions~\ref{lem:lemRedBKeyIsRed},~\ref{lem:lemDetAggMsg}, and Corollary~\ref{col:det-red}).

\subsection{In-Degrees}
The Spark GraphX library implements several graph algorithms using
aggregation. We show how our specification helps to understand and
analyze Spark programs utilizing aggregate combinators.

Let $G = (V, E)$ with $E
\subseteq V \times V$ be a directed graph. We define the \emph{in-degree} of a
vertex $v \in V$ as $| \{ (u, v) : (u, v) \in E \}|$. The GraphX
library uses the function \hscode{\kw{aggregateMessages}} to compute in-degrees of
vertices in a graph RDD. Consider the following {\specspark}
specification for the GraphX implementation:
\begin{lstlisting}[language=Haskell,basicstyle=\sffamily\footnotesize]
inDegrees graphRdd = 
  let sendMsg _ _ dst _ _ = [(dst, 1)]
  in aggregateMessages sendMsg (+) graphRdd
\end{lstlisting}
By our specification, \hscode{\kw{aggregateMessages}} invokes
\hscode{sendMsg} on every edge in \hscode{graphRdd}. The \hscode{sendMsg}
function sends the message \hscode{1} to the destination vertex of an
edge. If several messages are sent to a vertex, they are summed up.
Hence \hscode{inDegree} returns a vertex RDD where each
vertex has the number of its incoming edges as the attribute. They are
in-degrees of vertices in \hscode{graphRdd}.
The call to \hscode{\kw{aggregateMessages}} has a deterministic outcome because $(\hscode{\kw{Int}}, \hscode{(+)})$ is a commutative semigroup (Propositions~\ref{lem:lemRedBKeyIsRed},~\ref{lem:lemDetAggMsg}, and Corollary~\ref{col:det-red}).
\hide{
In the implementation, the directed graph $G$ is represented as a \hscode{GraphRDD G}. According to the formal specification in Haskell, the function \hscode{aggregateMessages} first applies \hscode{sendMsg} to all edges $(v,v')\in E$ and produce a Pair RDD with the set of elements $P=\{(v,1): (v',v)\in E\}$. 
Then it applies \hscode{reduceByKey} with the function \hscode{(+)} to
the Pair RDD, which labels each vertex $v$ with $\Sigma_{(v,1)\in
  P}1$, which actually compute the number $\Sigma_{(v',v)\in
  E}1=|\{(v',v):(v',v)\in E\}|$.
}


\newpage

\changepage{50mm}{60mm}{-30mm}{-30mm}{}{-25mm}{}{}{}


\vspace{-0.0mm}
\section{Missing Proofs}\label{app:proofs}
\vspace{-0.0mm}

We start with proving the following auxiliary lemma.

\begin{lemma}\label{lem:broken-foldl}
\begin{equation}
  \foldl(f, \zero, \li_1 \cat \li_2) = \foldl(f, \foldl(f, \zero, \li_1), \li_2)
\end{equation}
\end{lemma}

\begin{proof}
By induction on the length of $\li_1$.
\begin{itemize}
  \item  for $\li_1 = [~]$:
    \begin{align*}
      \foldl(f, \foldl(f, \zero, [~]), \li_2)
      &= \foldl(f, \zero, \li_2)
      &\expl{def.\ of $\foldl$}
      \\
      &= \foldl(f, \zero, [~] \cat \li_2)
      &\expl{def.\ of $\cat$}
    \end{align*}

  \item  suppose the lemma holds for all $\li_1$ of length $n$.
    Now consider the list $x:\li_1$.
    It follows that
    \begin{align*}
      \foldl(f, \zero, x:p_1 \cat p_2)
      &= \foldl(f, f(\zero, x), p_1 \cat p_2)
      &\expl{def.\ of $\foldl$}
      \\
      &= \foldl(f, \foldl(f, f(\zero, x), p_1), p_2)
      &\expl{IH}
      \\
      &= \foldl(f, \foldl(f, \zero, x:p_1), p_2)
      &\expl{def.\ of $\foldl$} & ~\qed
    \end{align*}
\end{itemize}
\end{proof}


In the following we use the following function:
\begin{lstlisting}[language=Haskell,deletekeywords={seq},basicstyle=\sffamily\footnotesize,morekeywords={reducel,aggregateList,aggregate},escapechar=@]
aggregateList part  z  seq  comb xs = aggregate@$\dt$@ z seq comb (part xs)
\end{lstlisting}

\begin{restatable}{lemma}{aggregateCommu}\label{lemma:aggregate_commu}
The following are necessary (though not sufficient) conditions for a~call
$\agg(\zero, \seq, \comboper, \partit(\datalist))$ to be deterministic:
\begin{enumerate}
  \item  $\zero$ is the identity of $\comboper$ on $\gamma = \imgof{\foldl(\seq,
    \zero)}$,
  \item  $\comboper$ is closed on $\gamma$,
  \item  $\comboper$ is commutative on $\gamma$, and
  \item  $\comboper$ is associative on $\gamma$.
\end{enumerate}
\end{restatable}

\begin{proof}
\begin{enumerate}

  \item  We assume that $\agg(\zero, \seq, \comboper, \partit(\datalist))$ is
    deterministic and show that $\zero$ is both the left and the right identity
    of $\comboper$ on $\gamma$.
    First, assume the following partitioning: $\partit_1(\datalist) =
    [\datalist]$.
    From the assumption that the $\agg$ is deterministic, it follows that
    \begin{align*}
      \foldlzof{\datalist} &= \aggList(\partit_1, \zero, \seq, \comboper,
      \datalist) \\
      &= \foldl(\comboper, \zero, [\foldlzof{L}])
      &\expl{def.\ of $\aggList$} \\
      &= \foldl(\comboper, \zero \comboper \foldlzof{L}, [])
      &\expl{def.\ of $\foldl$} \\
      &= \zero \comboper \foldlzof{L}
      &\expl{def.\ of $\foldl$}
    \end{align*}
    Therefore, $\zero$ is the left identity of $\comboper$ on $\gamma$.

    Second, assume the following partitioning: $\partit_2(\datalist) = [L,
    []]$.
    From the assumption that the $\agg$ is deterministic, it follows that
    \begin{align*}
      \foldlzof{\datalist} &= \aggList(\partit_2, \zero, \seq, \comboper,
      \datalist) \\
      &= \foldl(\comboper, \zero, [\foldlzof{L}, \foldlzof{[]}])
      &\expl{def.\ of $\aggList$} \\
      &= \foldl(\comboper, \zero, [\foldlzof{L}, \zero])
      &\expl{def.\ of $\foldlzof{\cdot}$ and $\foldl$} \\
      &= \foldl(\comboper, \zero \comboper \foldlzof{L}, [\zero])
      &\expl{def.\ of $\foldl$} \\
      &= \foldl(\comboper, \foldlzof{L}, [\zero])
      &\expl{$\zero$ is the left id.\ of $\comboper$} \\
      &= \foldl(\comboper, \foldlzof{L} \comboper \zero, [])
      &\expl{def.\ of $\foldl$} \\
      &= \foldlzof{L} \comboper \zero
      &\expl{def.\ of $\foldl$}
    \end{align*}
    Therefore, $\zero$ is also the right identity of $\comboper$ on $\gamma$.

  \item  We assume that $\agg(\zero, \seq, \comboper, \rdd(\datalist))$ is
    deterministic and show that $\comboper$ is closed on $\gamma$.
    First, we assume that $\datalist = p_1 \cat p_2$ and consider the following
    partitioning: $\partit(p_1 \cat p_2) = [p_1, p_2]$.
    From the assumption that the $\agg$ is deterministic, it follows that
    \begin{align*}
    \foldlzof{p_1 \cat p_2}
    &= \aggList(\partit, \zero, \seq, \comboper, \datalist)
    \\
    &= \foldl(\comboper, \zero, [\foldlzof{p_1}, \foldlzof{p_2}])
    &\expl{def.\ of $\aggList$}\\
    &= \foldl(\comboper, \zero \comboper \foldlzof{p_1}, [\foldlzof{p_2}])
    &\expl{def.\ of $\foldl$}\\
    &= \foldl(\comboper, \foldlzof{p_1}, [\foldlzof{p_2}])
    &\expl{$\zero$ is the id.\ of $\comboper$}\\
    &= \foldl(\comboper, \foldlzof{p_1} \comboper\foldlzof{p_2}, [])
    &\expl{def.\ of $\foldl$}\\
    &= \foldlzof{p_1} \comboper\foldlzof{p_2}
    &\expl{def.\ of $\foldl$}
    \end{align*}
    Therefore $\comboper$ is closed on $\gamma$.

  \item  We assume that $\agg(\zero, \seq, \comboper, \rdd(\datalist))$ is
    deterministic and show that $\comboper$ is commutative on $\gamma$.
    First, we assume that $\datalist = p_1 \cat p_2$ and consider the following
    two partitionings: $\partit_1(p_1 \cat p_2) = [p_1, p_2]$ and
    $\partit_2(p_1 \cat p_2) = [p_2, p_1])$.
    From the assumption that the $\agg$ is deterministic, it follows that
    \begin{align*}
      &&\aggList(\partit_1, \zero, \seq, \comboper, \datalist) &= \aggList(\partit_2, \zero, \seq, \comboper, \datalist) \\
    \iff&&\foldl(\comboper, \zero, [\foldlzof{p_1}, \foldlzof{p_2}]) &= \foldl(\comboper, \zero, [\foldlzof{p_2}, \foldlzof{p_1}])
      &\expl{def.\ of $\aggList$} \\
    \iff&&\foldl(\comboper, \zero \comboper \foldlzof{p_1}, [\foldlzof{p_2}]) &= \foldl(\comboper, \zero \comboper \foldlzof{p_2}, [\foldlzof{p_1}])
      &\expl{def.\ of $\foldl$} \\
    \iff&&\foldl(\comboper, \foldlzof{p_1}, [\foldlzof{p_2}]) &= \foldl(\comboper, \foldlzof{p_2}, [\foldlzof{p_1}])
      &\expl{$\zero$ is the id.\ of $\comboper$} \\
    \iff&&\foldl(\comboper, \foldlzof{p_1} \comboper \foldlzof{p_2}, []) &= \foldl(\comboper, \foldlzof{p_2} \comboper \foldlzof{p_1}, [])
      &\expl{def.\ of $\foldl$} \\
    \iff&&\foldlzof{p_1} \comboper \foldlzof{p_2} &= \foldlzof{p_2} \comboper \foldlzof{p_1})
      &\expl{def.\ of $\foldl$}
    \end{align*}
    Therefore, $\comboper$ is commutative on $\gamma$.

  \item  We assume that $\agg(\zero, \seq, \comboper, \rdd(\datalist))$ is
    deterministic and show that $\comboper$ is associative on $\gamma$.
    First, we assume that $\datalist = p_1 \cat p_2 \cat p_3$ and consider the
    following two partitionings: $\partit_1(p_1 \cat p_2 \cat p_3) = [p_1, p_2,
    p_3]$ and
    $\partit_2(p_1 \cat p_2 \cat p_3) = [p_2, p_3, p_1])$.
    From the assumption that the $\agg$ is deterministic, it follows that
    \begin{align*}
      &&\aggList(\partit_1, \zero, \seq, \comboper, \datalist) &= \aggList(\partit_2, \zero, \seq, \comboper, \datalist) \\
    \iff&&\foldl(\comboper, \zero, [\foldlzof{p_1}, \foldlzof{p_2}, \foldlzof{p_3}]) &= \foldl(\comboper, \zero, [\foldlzof{p_2}, \foldlzof{p_3}, \foldlzof{p_1}])
      &\expl{def.\ of $\aggList$} \\
    \iff&&\foldl(\comboper, \zero \comboper \foldlzof{p_1}, [\foldlzof{p_2}, \foldlzof{p_3}]) &= \foldl(\comboper, \zero \comboper \foldlzof{p_2}, [\foldlzof{p_3}, \foldlzof{p_1}])
      &\expl{def.\ of $\foldl$} \\
    \iff&&\foldl(\comboper, \foldlzof{p_1}, [\foldlzof{p_2}, \foldlzof{p_3}]) &= \foldl(\comboper, \foldlzof{p_2}, [\foldlzof{p_3}, \foldlzof{p_1}])
      &\expl{$\zero$ is the id.\ of $\comboper$} \\
    \iff&&\foldl(\comboper, \foldlzof{p_1} \comboper \foldlzof{p_2}, [\foldlzof{p_3}]) &= \foldl(\comboper, \foldlzof{p_2} \comboper \foldlzof{p_3}, [\foldlzof{p_1}])
      &\expl{def.\ of $\foldl$} \\
    \iff&&\foldl(\comboper, (\foldlzof{p_1} \comboper \foldlzof{p_2}) \comboper \foldlzof{p_3}, []) &= \foldl(\comboper, (\foldlzof{p_2} \comboper \foldlzof{p_3}) \comboper \foldlzof{p_1}, [])
      &\expl{def.\ of $\foldl$} \\
    \iff&&(\foldlzof{p_1} \comboper \foldlzof{p_2}) \comboper \foldlzof{p_3} &= (\foldlzof{p_2} \comboper \foldlzof{p_3}) \comboper \foldlzof{p_1}
      &\expl{def.\ of $\foldl$} \\
    \iff&&(\foldlzof{p_1} \comboper \foldlzof{p_2}) \comboper \foldlzof{p_3} &= \foldlzof{p_1} \comboper (\foldlzof{p_2} \comboper \foldlzof{p_3})
      &\expl{comm.\ of $\comboper$}
    \end{align*}
    Therefore, $\comboper$ is associative on $\gamma$.\qed
\end{enumerate}
\end{proof}

\begin{lemma}\label{lem:homo-means-good}
For all functions $h: \listof{A} \to B$, the following are equivalent:
\begin{enumerate}
  \item  $h$ is a~list homomorphism to $(B, \odot, \bot)$,
  \item  $\forall \xss \in \listof{\listof{A}}: \foldl(\odot, \bot, \map(h, \xss)) = h(\concat(\xss))$.
\end{enumerate}
\end{lemma}

\begin{proof}
\begin{itemize}
  \item[$(1 \Rightarrow 2)$:]
    By induction on the length of $\xss$:
    \begin{itemize}
      \item  for $\xss = [~]$:
        \begin{align*}
        \foldl(\odot, \bot, \map(h, [~]))
        &= \foldl(\odot, \bot, [~])
        &\expl{def.\ of $\map$} \\
        &= \bot
        &\expl{def.\ of $\foldl$} \\
        &= h([~])
        &\expl{assumption} \\
        &= h(\concat([~]))
        &\expl{def.\ of $\concat$} \\
        \end{align*}
    \item  Consider the following induction hypothesis for $\xss_n$ of the length $n$:
    \begin{equation}
      \mathrm{IH}: \foldl(\odot, \bot, \map(h, \xss_n)) = h(\concat(\xss_n)).
    \end{equation}
    For $\xss_n \cat [\xs]$ we proceed as follows:
    \begin{align*}
      \foldl(\odot, \bot, \map(h, \xss_n \cat [\xs]))
      &=
      \foldl(\odot, \bot, \map(h, \xss_n) \cat \map(h, [\xs]))
      &\expl{def.\ of $\map$}\\
      &=
      \foldl(\odot, \foldl(\odot, \bot, \map(h, \xss_n)), \map(h, [\xs]))
      &\expl{Lemma~\ref{lem:broken-foldl}}\\
      &=
      \foldl(\odot, \foldl(\odot, \bot, \map(h, \xss_n)), [h(\xs)])
      &\expl{def.\ of $\map$}\\
      &=
      \foldl(\odot, h(\concat(\xss_n)), [h(\xs)])
      &\expl{IH}\\
      &=
      h(\concat(\xss_n)) \odot h(\xs)
      &\expl{def.\ of $\foldl$}\\
      &=
      h(\concat(\xss_n) \cat \xs)
      &\expl{assumption}\\
      &=
      h(\concat(\xss_n \cat [\xs]))
      &\expl{def.\ of $\concat$}\\
    \end{align*}
    \end{itemize}

  \item[$(2 \Rightarrow 1)$:]
    We prove that the two properties of a~list homomorphism hold:
    \begin{itemize}
      \item  From $\foldl(\odot, \bot, \map(h, [~])) = h(\concat([~]))$ it
        follows that $h([~]) = \bot$.

      \item  To prove that $h(\xs \cat \ys) = h(\xs) \odot h(\ys)$, first we
        consider the list $\xss = [\xs]$:
        \begin{align*}
        &&
        \foldl(\odot, \bot, \map(h, [\xs])) &= h(\concat([\xs]))
        \\
        \iff&&
        \foldl(\odot, \bot, [h(\xs)]) &= h(\xs)
        &\expl{def.\ of $\map$, def.\ of $\concat$} \\
        \iff&&
        \foldl(\odot, \bot \odot h(\xs), [~]) &= h(\xs)
        &\expl{def.\ of $\foldl$} \\
        \iff&&
        \bot \odot h(\xs) &= h(\xs)
        \numberthis
        \label{067cdbbe-6c92-4514-8212-60c179160a01}
        &\expl{def.\ of $\foldl$} \\
        \end{align*}
        Then we consider the list $\xss = [\xs, \ys]$:
        \begin{align*}
        &&
        \foldl(\odot, \bot, \map(h, [\xs, \ys])) &= h(\concat([\xs, \ys]))
        \\
        \iff&&
        \foldl(\odot, \bot, [h(\xs), h(\ys)]) &= h(\xs \cat \ys)
        &\expl{def.\ of $\map$, def.\ of $\concat$} \\
        \iff&&
        \foldl(\odot, \bot \odot h(\xs), [h(\ys)]) &= h(\xs \cat \ys)
        &\expl{def.\ of $\foldl$} \\
        \iff&&
        \foldl(\odot, (\bot \odot h(\xs)) \odot h(\ys), [~]) &= h(\xs \cat \ys)
        &\expl{def.\ of $\foldl$} \\
        \iff&&
        (\bot \odot h(\xs)) \odot h(\ys) &= h(\xs \cat \ys)
        &\expl{def.\ of $\foldl$} \\
        \iff&&
        h(\xs) \odot h(\ys) &= h(\xs \cat \ys)
        &\expl{(\ref{067cdbbe-6c92-4514-8212-60c179160a01})}
        &&
    ~\qed
        \end{align*}
    \end{itemize}
\end{itemize}
\end{proof}


\homoIff*

\begin{proof}
\begin{enumerate}
  \item[$\Rightarrow$:]
    \begin{enumerate}
      \item  Proving 1:
        Follows from Lemma~\ref{lemma:aggregate_commu}.

    \item  Proving 2:
    consider the list $\xs \cat \ys$ and its partitioning
      $\partit(\xs \cat \ys) = [\xs, \ys]$.
      \begin{align*}
      &&
      \aggList(\partit, \zero, \seq, \comboper, \xs \cat \ys)
      &=
      \foldlzof{\xs \cat \ys}
      &\expl{def.\ of det.\ $\agg$}
      \\
      \iff&&
      \foldl(\comboper, \zero, [\foldlzof{\xs}, \foldlzof{\ys}])
      &=
      \foldlzof{\xs \cat \ys}
      &\expl{def.\ of $\aggList$}
      \\
      \iff&&
      \foldl(\comboper, \zero \comboper \foldlzof{\xs}, [\foldlzof{\ys}])
      &=
      \foldlzof{\xs \cat \ys}
      &\expl{def.\ of $\foldl$}
      \\
      \iff&&
      \foldl(\comboper, \foldlzof{\xs}, [\foldlzof{\ys}])
      &=
      \foldlzof{\xs \cat \ys}
      &\expl{$\zero$ is the id.\ of $\comboper$}
      \\
      \iff&&
      \foldl(\comboper, \foldlzof{\xs} \comboper \foldlzof{\ys}, [~])
      &=
      \foldlzof{\xs \cat \ys}
      &\expl{def.\ of $\foldl$}
      \\
      \iff&&
      \foldlzof{\xs} \comboper \foldlzof{\ys}
      &=
      \foldlzof{\xs \cat \ys}
      &\expl{def.\ of $\foldl$}
      \end{align*}
    \end{enumerate}

  \item[$\Leftarrow$:]
    Consider an arbitrary partitioning $\partit(\datalist)$ of $\datalist$ and
    its permutation $\perm$ s.t.\ $\datalist =
    \concat(\perm(\partit(\datalist)))$.
    From the definition of $\foldlzof{\cdot}$, it follows that
    $\foldlzof{[~]} = \foldl(\seq, \zero, [~]) = \zero$, and, therefore,
    $\foldlzof{\cdot}$ is a~list homomorphism to $(\imgof{\foldl(\seq,
    \zero)}, \comboper, \zero)$.
    From Lemma~\ref{lem:homo-means-good} it follows
    that
    \begin{align*}
      &&
      \foldl(\comboper, \zero, \map(\foldlzof{\cdot}, \perm(\partit(\datalist))))
      &=
      \foldlzof{\concat(\perm(\partit(\datalist)))}
      \\
      \iff &&
      \foldl(\comboper, \zero, \map(\foldlzof{\cdot}, \perm(\partit(\datalist))))
      &=
      \foldlzof{\datalist}
      &\expl{def.\ of $\perm$ and $\partit$}
      \\
      \iff &&
      \aggList(\perm \circ \partit, \zero, \seq, \comboper, \datalist)
      &=
      \foldlzof{\datalist}
      &\expl{def.\ of $\aggList$}
    \end{align*}
    Because $\comboper$ is associative and commutative, it follows that
    $\aggList(\perm_x \circ \partit, \zero, \seq, \comboper, \datalist) =
    \foldlzof{\datalist}$ for any $\perm_x$.
    Therefore, $\agg(\zero, \seq, \comboper, \rddof{\datalist})$ is
    deterministic.  \qed
\end{enumerate}
\end{proof}


\aggreSuffCond*

\begin{proof}
\begin{itemize}
  \item[$1 \implies 2$:]
    This is a special case.
    We pick $p_1$ such that $\foldlzof{p_1} = e$ and $p_2 = [d]$.
    When we substitute into~(\ref{eq:9461f388-4968-41b8-83af-b8a49f9cfedc}), we get
    \begin{equation}\label{e06e34ec-d86e-41ef-a2f5-83a58e0ffa16}
    \foldlzof{p_1 \cat [d]} = e \comboper \foldlzof{[d]} .
    \end{equation}
    For the left-hand side, according to Lemma~\ref{lem:broken-foldl}, it holds
    that
    \begin{equation}
    \foldlzof{p_1 \cat [d]}
    = \foldl(\seq, \zero, p_1 \cat [d])
    = \foldl(\seq, \foldl(\seq, \zero, p_1), [d])
    = \foldl(\seq, \foldlzof{p_1}, [d]) .
    \end{equation}

    After substitution, we get $\foldl(\seq, e, [d])$, which is (from the
    definition of $\foldl$) equal to $\seq(e, d)$.
    For the right-hand side of~(\ref{e06e34ec-d86e-41ef-a2f5-83a58e0ffa16}), we
    just notice that $\foldlzof{[d]} = \foldl(\seq, \zero, [d]) = \seq(\zero,
    d)$.

  \item[$2 \implies 1$:]
    Set $x = \foldl(\seq, \zero, p_1) = \foldlzof{p_1}$ and substitute
    into~(\ref{eq:9461f388-4968-41b8-83af-b8a49f9cfedc}) to obtain a~new target
    for proving:
    \begin{align}
    && \foldlzof{p_1 \cat p_2}
    &= \foldlzof{p_1} \comboper \foldlzof{p_2}
    \nonumber \\
    \iff
    && \foldl(\seq, \zero, p_1 \cat p_2)
    &= \foldlzof{p_1} \comboper \foldlzof{p_2}
    &\expl{def.\ of $\foldlzof{\cdot}$}
    \nonumber \\
    \iff
    && \foldl(\seq, \foldl(\seq, \zero, p_1), p_2)
    &= \foldlzof{p_1} \comboper \foldlzof{p_2}
    & \expl{Lemma~\ref{lem:broken-foldl}} \nonumber \\
    \iff
    && \foldl(x, \seq, p_2)
    &= x \comboper \foldlzof{p_2}
    & \expl{subst.\ of $x$} \label{eq:aslkfjlkajsdalkfs}
    \end{align}
    We prove~(\ref{eq:aslkfjlkajsdalkfs}) using induction on the length $n$ of
    $p_2$.
    \begin{itemize}
    \item[$n = 0$:]
      for $p_2 = []$, we get to prove the following:
      \begin{equation}
      \foldl(\seq, x, []) = x \comboper \foldl(\seq, \zero, [])).
      \end{equation}
      From the definition of $\foldl$, we get an equivalent formula
      \begin{equation}
        x = x \comboper \zero ,
      \end{equation}
      which is true due to $\zero$ being the identity of $\comboper$ on
      $\gamma$.

    \item[$n = i+1$:]  We assume~(\ref{eq:aslkfjlkajsdalkfs}) holds for $p_2$ of
      length $i$, i.e.
      \begin{equation}
      \text{IH}:~~\foldl(\seq, x, p_i) = x \comboper \foldl(\seq, \zero, p_i)
      \end{equation}
      and prove that, for any $h \in \alpha$,
      \begin{equation}
        \foldl(\seq, x, p_i \cat [h]) = x \comboper \foldl(\seq, \zero, p_i
        \cat [h]) .
      \end{equation}
      We do it in the following way:
      \begin{align*}
        \MoveEqLeft
        \foldl(\seq, x, p_i \cat [h]) \\
        & = \foldl(\seq, \foldl(\seq, x, p_i), [h]) & \expl{Lemma~\ref{lem:broken-foldl}} \\
        & = \foldl(\seq, \seq(\foldl(\seq, x, p_i), h), []) & \expl{def.\ of $\foldl$} \\
        & = \seq(\foldl(\seq, x, p_i), h) & \expl{def.\ of $\foldl$} \\
        & = \foldl(\seq, x, p_i) \comboper \seq(\zero, h) & \expl{appl.\ of~(\ref{eq:cd2c544e-7b22-4c31-b1d0-1160239694b7})} \\
        & = (x \comboper \foldl(\seq, \zero, p_i)) \comboper \seq(\zero, h) & \expl{IH} \\
        & = x \comboper (\foldl(\zero, \seq, p_i) \comboper \seq(\zero, h)) & \expl{assoc.\ of $\comboper$} \\
        & = x \comboper \seq(\foldl(\seq, \zero, p_i),  h) & \expl{appl.\ of~(\ref{eq:cd2c544e-7b22-4c31-b1d0-1160239694b7})} \\
        & = x \comboper \foldl(\seq, \seq(\foldl(\seq, \zero, p_i),  h), []) & \expl{def.\ of $\foldl$} \\
        & = x \comboper \foldl(\seq, \foldl(\seq, \zero, p_i), [h])) & \expl{def.\ of $\foldl$} \\
        & = x \comboper \foldl(\seq, \zero, p_i \cat [h])) & \expl{Lemma~\ref{lem:broken-foldl}}
      &&~\qed
      \end{align*}
    \end{itemize}
\end{itemize}
\end{proof}


\begin{lemma}\label{lem:redl}
\begin{equation}
  \redl(f, \xs) = \redl'(f, \xs)
\end{equation}
where
\begin{lstlisting}[language=Haskell,basicstyle=\sffamily\upshape\footnotesize,morekeywords={reducel}]
reducel' f xs = fromJust (foldl f' Nothing xs)
  where f' x y = case x of
    Nothing $\rightarrow$ Just y
    Just x' $\rightarrow$ Just (f x' y)
\end{lstlisting}
\end{lemma}

\begin{proof}
by induction on the length of $\xs$:
\begin{enumerate}
  \item  for $\xs = [~]$, both $\red$ and $\redl$ are undefined.

  \item  for $\xs = [x]$:
    \begin{align*}
    \redl(f, [x])  &= \foldl(f, x, [~]) = x
    \end{align*}
    and
    \begin{align*}
    \redl'(f, [x]) &= \fromJust(\foldl(f', \nothing, [x]) & \expl{def.\ of $\redl'$} \\
                   &= \fromJust(\foldl(f', f'(\nothing, x), [~]) &\expl{def.\ of $\foldl$} \\
                   &= \fromJust(\foldl(f', \just(x), [~]) &\expl{def.\ of $f'$} \\
                   &= \fromJust(\just(x)) &\expl{def.\ of $\foldl$} \\
                   &= x &\expl{def.\ of $\fromJust$}
    \end{align*}

  \item  assume the following induction hypothesis:
    \begin{equation}
      \redl(f, x:\xs) = \redl'(f', x:\xs) = R
    \end{equation}
    We now prove that the lemma holds for $x:\xs \cat [a]$.
    First, we compute the result for $\redl(f, x:\xs \cat [a])$:
    \begin{align*}
    \redl(f, x:\xs \cat [a])
    &=
    \foldl(f, x, \xs \cat [a])
    &\expl{def.\ of $\redl$}
    \\
    &=
    \foldl(f, \foldl(f, x, \xs), [a])
    &\expl{Lemma~\ref{lem:broken-foldl}}
    \\
    &=
    \foldl(f, \redl(f, x:\xs), [a])
    &\expl{def.\ of $\redl$}
    \\
    &=
    \foldl(f, R, [a])
    &\expl{IH}
    \\
    &=
    \foldl(f, f(R, a), [~])
    &\expl{def.\ of $\foldl$}
    \\
    &=
    f(R, a)
    &\expl{def.\ of $\foldl$}
    \end{align*}
    We proceed by computing the result for $\redl'(f, x:\xs \cat [a])$:
    \begin{align*}
    \MoveEqLeft
    \redl'(f, x:\xs \cat [a]) \\
    {}={}&
    \fromJust(\foldl(f', \nothing, x:\xs \cat [a]))
    &\expl{def.\ of $\redl'$}
    \\
    {}={}&
    \fromJust(\foldl(f', \foldl(f', \nothing, x:\xs),  [a]))
    &\expl{Lemma~\ref{lem:broken-foldl}}
    \\
    {}={}&
    \fromJust(\foldl(f', f'(\foldl(f', \nothing, x:\xs), a), [~]))
    &\expl{def.\ of $\foldl$}
    \\
    {}={}&
    \fromJust(f'(\foldl(f', \nothing, x:\xs), a))
    &\expl{def.\ of $\foldl$}
    \\
    &
    \langle \text{$f'$ is applied at least once on $x:\xs$ $\implies$ the result of the nested $\foldl$ cannot be $\nothing$} \rangle
    \\
    {}={}&
    \fromJust(\just(f(\fromJust(\foldl(f', \nothing, x:\xs)), a))
    &\expl{def.\ of $f'$}
    \\
    {}={}&
    f(\fromJust(\foldl(f', \nothing, x:\xs)), a)
    &\expl{def.\ of $\fromJust$}
    \\
    {}={}&
    f(\redl'(f',x:\xs), a)
    &\expl{def.\ of $\redl'$}
    \\
    {}={}&
    f(R, a)
    &\expl{IH}
    &&~\qed
    \end{align*}
\end{enumerate}
\end{proof}


\lemmaRedIsAggr*

\begin{proof}
We show that given the following definition of the function $\red''$,
\begin{lstlisting}[language=Haskell,basicstyle=\sffamily\footnotesize,morekeywords={aggregate,reduce},deletekeywords={seq},mathescape]
reduce'' :: ($\alpha$ $\rightarrow$ $\alpha$ $\rightarrow$ $\alpha$) $\rightarrow$ RDD $\alpha$ $\rightarrow$ $\alpha$
reduce'' ($\comboper$) rdd = fromJust (aggregate Nothing seq' ($\comboper$') rdd) ,
\end{lstlisting}
it holds that $\red''(\comboper, \rdd) = \red\dt(\comboper, \rdd)$ for all
$\comboper$ and $\rdd$.
In case $\rdd$ is a~partitioning of an empty list, the result of both $\red'$
and $\red''$ is undefined.
For a~non-empty list:
\begin{align*}
  \MoveEqLeft
  \red''(\comboper', \xs:\xss)
  \\
  {}={}&
  \fromJust(\agg(\nothing, \seq', \comboper', \xs:\xss))
  &\expl{def.\ of $\red''$}
  \\
  {}={}&
  \fromJust(\foldl(\comboper', \nothing, \map(\lambda \ys\;.\;\foldl(\seq', \nothing, \ys), \xs:\xss)))
  &\expl{def.\ of $\agg$}
  \\
  & \text{$\langle$ from the assumption on partitionings, no element of $\xs:\xss$ is empty $\rangle$}
  \\
  {}={}&
  \fromJust(\foldl(\comboper', \nothing, \map(\lambda \ys\;.\;\just(\fromJust(\foldl(\seq', \nothing, \ys))), \xs:\xss)))
  &\expl{def.\ of $\fromJust$}
  \\
  {}={}&
  \fromJust(\foldl(\comboper', \nothing, \map(\lambda \ys\;.\;\just(\redl(\comboper, \ys)), \xs:\xss)))
  &\expl{Lemma~\ref{lem:redl}}
  \\
  {}={}&
  \fromJust(\foldl(\comboper', \nothing, \just(\redl(\comboper, \xs)): \map(\lambda \ys\;.\;\just(\redl(\comboper, \ys)), \xss)))
  &\expl{def.\ of $\map$}
  \\
  {}={}&
  \fromJust(\foldl(\comboper', \nothing \comboper' \just(\redl(\comboper, \xs)), \map(\lambda \ys\;.\;\just(\redl(\comboper, \ys)), \xss)))
  &\expl{def.\ of $\foldl$}
  \\
  {}={}&
  \fromJust(\foldl(\comboper', \just(\redl(\comboper, \xs)), \map(\lambda \ys\;.\;\just(\redl(\comboper, \ys)), \xss)))
  &\expl{def.\ of $\comboper'$}
  \\
  %
  {}={}&
  \fromJust(\just(\foldl(\comboper, \redl(\comboper, \xs), \map(\lambda \ys\;.\;\redl(\comboper, \ys), \xss))))
  &\expl{def.\ of $\comboper'$}
  \\
  {}={}&
  \foldl(\comboper, \redl(\comboper, \xs), \map(\lambda \ys\;.\;\redl(\comboper, \ys), \xss))
  &\expl{def.\ of $\fromJust$}
  \\
\hide{
  {}={}&
  \foldl(\comboper, \redl(\comboper, \xs), \map(\lambda \ys\;.\;\redl(\comboper, \ys), \xss))
  &\expl{def.\ of $\fromJust$}
  \\
}
  {}={}&
  \redl(\comboper,\redl(\comboper, \xs):\map(\lambda \ys\;.\;\redl(\comboper, \ys), \xss))
  &\expl{def.\ of $\redl$}
  \\
  {}={}&
  \redl(\comboper, \map(\lambda \ys\;.\;\redl(\comboper, \ys), \xs:\xss))
  &\expl{def.\ of $\map$}
  \\
  {}={}&
  \red\dt(\comboper, \xs:\xss)
  &\expl{def.\ of $\red\dt$}
  &&~ \qed
\end{align*}
\end{proof}


\coroReduce*

\begin{proof}
From Lemma~\ref{lem:lemmaRedIsAggr}, it follows that we can investigate the
function $\agg(\nothing, \seq', \comboper', \rdd)$ instead of $\red(\comboper,
\rdd)$.
From Corollary~\ref{col:det-agg}, we obtain that $\agg(\nothing, \seq',
\comboper', \rdd)$ has deterministic outcome iff the following two conditions hold:
\begin{enumerate}
  \item  $(\imgof{\foldl(\seq', \nothing)}, \comboper', \nothing)$ is a commutative monoid,
  \item  $\forall d \in \alpha, e \in \imgof{\foldl(\seq', \nothing)}: \seq'(e, d) = e \comboper' \seq'(\nothing, d)$.
\end{enumerate}
We start with investigating condition~2:
\begin{itemize}
  \item  For the case $e = \nothing$:
  \begin{align*}
    &&
    \seq'(e, d)
    &=
    e \comboper' \seq'(\nothing, d)
    \\
    \iff&&
    \seq'(\nothing, d)
    &=
    \nothing \comboper' \seq'(\nothing, d)
    &\expl{subst.\ of $e = \nothing$}
    \\
    \iff&&
    \just(d)
    &=
    \nothing \comboper' \just(d)
    &\expl{def.\ of $\seq'$}
    \\
    \iff&&
    \just(d)
    &=
    \just(d)
    &\expl{def.\ of $\comboper'$}
  \end{align*}

  \item  For the case $e = \just(x)$:
  \begin{align*}
    &&
    \seq'(e, d)
    &=
    e \comboper' \seq'(\nothing, d)
    \\
    \iff&&
    \seq'(\just(x), d)
    &=
    \just(x) \comboper' \seq'(\nothing, d)
    &\expl{subst.\ of $e = \just(x)$}
    \\
    \iff&&
    \just(x \comboper d)
    &=
    \just(x) \comboper' \just(d)
    &\expl{def.\ of $\seq'$}
    \\
    \iff&&
    \just(x \comboper d)
    &=
    \just(x \comboper d)
    &\expl{def.\ of $\comboper'$}
  \end{align*}
\end{itemize}

We can observe that the condition is a~tautology.
Therefore, the condition~1 is a~sufficient and necessary condition for a~call
to $\agg(\nothing, \seq', \comboper', \rdd)$ to have a~deterministic outcome.

We proceed by investigating the conditions for $(\imgof{\foldl(\seq',
\nothing)}, \comboper', \nothing)$ to be a~commutative monoid.
First, we observe that for $\comboper: \alpha \times \alpha \to \alpha$, it
holds that $\imgof{\foldl(\seq', \nothing)} = \maybe(\alpha)$.
\begin{itemize}
  \item  \emph{Identity}:  From the definition, $\nothing$ is the identity of
    $\comboper'$.

  \item  \emph{Commutativity}:  From the definition, $\comboper'$ is
    commutative iff $\comboper$ is commutative.

  \item  \emph{Associativity}:  Consider elements $a,b,c \in
    \maybe(\alpha)$.
    We explore when $(a \comboper' b) \comboper' c = a \comboper' (b \comboper' c)$:
    \begin{itemize}

      \item  If any member of $\{a, b, c\}$ is $\nothing$, the condition holds
        because $\nothing$ is the (left and right) identity of $\comboper'$.

      \item  For $a = \just(a')$, $b = \just(b')$, and $c = \just(c')$, it
        holds that:
        \begin{align*}
          &&
          (\just(a) \comboper' \just(b)) \comboper' \just(c)
          &=
          \just(a) \comboper' (\just(b) \comboper' \just(c))
          \\
          \iff&&
          \just(a \comboper b) \comboper' \just(c)
          &=
          \just(a) \comboper' \just(b \comboper c)
          &\expl{def.\ of $\comboper'$}
          \\
          \iff&&
          \just((a \comboper b) \comboper c)
          &=
          \just(a \comboper (b \comboper c))
          &\expl{def.\ of $\comboper'$}
        \end{align*}
    \end{itemize}
    Therefore, $\comboper'$ is associative iff $\comboper$ is associative.

  \item  \emph{Closed}:
    It is easy to observe that $\comboper'$ is closed on $\maybe(\alpha)$.
\end{itemize}

From the previous conditions, we infer that $\agg(\nothing, \seq', \comboper',
\rdd)$ has deterministic outcome iff $(\alpha, \comboper)$ is a~commutative
semiring.
\qed
\end{proof}


\lemmaTreeAgg*

\begin{proof}
\begin{itemize}
  \item[$\Rightarrow$:]
Consider the following function:
\begin{lstlisting}[language=Haskell,basicstyle=\sffamily\footnotesize,morekeywords={dividel}]
dividel :: [$\alpha$] $\rightarrow$ ([$\alpha$], $\alpha$, $\alpha$, [$\alpha$])
dividel x1:x2:xs = ([], x1, x2, xs) .
\end{lstlisting}
Obviously, \hscode{dividel} is one possible way how \hscode{\kw{divide}!} can
function.
We further consider the following modification of \hscode{apply}:
\begin{lstlisting}[language=Haskell,basicstyle=\sffamily\footnotesize,morekeywords={dividel,applyl}]
applyl :: ($\beta$ $\rightarrow$ $\beta$ $\rightarrow$ $\beta$) $\rightarrow$ [$\beta$] $\rightarrow$ $\beta$
applyl comb [r] = r
applyl comb [r, r'] = comb r r'
applyl comb rs = let (ls', l', r', rs') = dividel rs in applyl comb (ls' ++ [comb l' r'] ++ rs')
\end{lstlisting}
After inlinining \hscode{\kw{dividel}} to \hscode{\kw{applyl}}, we can modify it to obtain yet futher modification:
\begin{lstlisting}[language=Haskell,basicstyle=\sffamily\footnotesize,morekeywords={dividel,applyl}]
applyl' :: ($\beta$ $\rightarrow$ $\beta$ $\rightarrow$ $\beta$) $\rightarrow$ [$\beta$] $\rightarrow$ $\beta$
applyl' comb [r] = r
-- applyl' comb [r, r'] = comb r r'
applyl' comb r1:r2:rs = applyl' comb ((comb r1 r2):rs)
\end{lstlisting}
Note that the case for a list of length 2 is reduntant now.
Clearly it holds that \hscode{\kw{applyl}'(f, xs) = reducel(f, xs)}.
If we substitute \hscode{\kw{reducel}} for \hscode{\kw{apply}} in the definition of
\hscode{\kw{treeAggregate}}, and further use the property of a~partitioning that it
is never an empty list, we obtain the definition of \hscode{\kw{aggregate}}.

\item[$\Leftarrow$:]
From Lemma~\ref{lemma:aggregate_commu}, it follows that $\comboper$ is
associative and commutative.
Therefore, any sequence of \hscode{\kw{divide}!}-\hscode{\kw{apply}} operations in
\hscode{\kw{apply}} will yield the same outcome as if we consider the
(deterministic) \hscode{\kw{dividel}}.
\qed
\end{itemize}
\end{proof}


\lemTreeRedIsRed*

\begin{proof}
Follows the same structure as the proof of Proposition~\ref{lemma:treeAgg}.
\qed
\end{proof}


When inferring conditions for a deterministic outcome of the call to
\hscode{\kw{aggregateByKey}}, we make use of the following auxiliary function:
\begin{lstlisting}[language=Haskell,basicstyle=\sffamily\footnotesize,deletekeywords={seq},morekeywords={aggregate,aggregateWithKey}]
aggregateWithKey :: $\alpha$ $\rightarrow$ $\gamma$ $\rightarrow$ ($\gamma$ $\rightarrow$ $\beta$ $\rightarrow$ $\gamma$) $\rightarrow$ ($\gamma$ $\rightarrow$ $\gamma$ $\rightarrow$ $\gamma$) $\rightarrow$ PairRDD $\alpha$ $\beta$ $\rightarrow$ $\gamma$
aggregateWithKey k z seq comb pairRdd =
  let select p = key p == k
     vrdd = filter (not . null)
              (map ((map value) . (filter select)) pairRdd)
  in aggregate z seq comb vrdd
\end{lstlisting}
%
%
We also use the following version of \hscode{\kw{aggregateByKey}} with the
partitioning given explicitly:
\begin{lstlisting}[language=Haskell,basicstyle=\sffamily\footnotesize,morekeywords={aggregateByKey,aggregateListByKey}]
aggregateListByKey :: ([($\alpha$, $\beta$)] $\rightarrow$ [[($\alpha$, $\beta$)]]) $\rightarrow$ $\gamma$ $\rightarrow$ ($\gamma$$\rightarrow$$\beta$$\rightarrow$$\gamma$) 
                     $\rightarrow$ ($\gamma$$\rightarrow$$\gamma$$\rightarrow$$\gamma$) $\rightarrow$ [($\alpha$, $\beta$)] $\rightarrow$ PairRDD $\alpha$ $\gamma$
aggregateListByKey part z mergeComb mergeValue list = aggregateByKey z mergeComb mergeValue (part list)
\end{lstlisting}

\begin{restatable}{lemma}{lemAggBKeyIsAggWKey}\label{lem:lemAggBKeyIsAggWKey}
It holds that
\vspace{-1ex}
\begin{align*}
  \MoveEqLeft
  \lookUp(k, \aggBKey(\zero, \seq, \comboper, \prdd)) = \aggWKey(k, \zero, \seq, \comboper, \prdd)) ,
\end{align*}
\noindent where $\lookUp$ searches the first value with a given
key in an RDD:
\begin{align*}
& \lookUp(k, \xss) = \head_z(\concat(\map(\map(\val \circ \filterkey~k),
\xss))) ,
\end{align*}
and $\head_z$ returns $z$ when the input is empty.
\end{restatable}

\begin{proof}
{ 
  \newcommand{\xseq}[0]{\otimes}
  \newcommand{\xcomb}[0]{\oplus}

To avoid too many parentheses, we use curried functions for
the proof of this lemma. We need a number of additional lemmas.
The following property allows one to swap $\filterkey~k$ and
$\foldl~(\mergeBy~(\xcomb))~[~]$:
\begin{equation}
   \filterkey~k \circ \foldl~(\mergeBy~(\xcomb))~[~] =
   \foldl~(\mergeBy~(\xcomb))~[~] \circ \filterkey~k \mbox{.}
   \label{eq:filterkey-fold}
\end{equation}
The next property says that, given a key $k$ and a binary
operator $(\odot)$, filtering the list with $k$ and performing
$\foldl~(\mergeBy (\odot))~[~]$ gives you a single value:
\begin{equation}
  \head_z \circ \map~\val \circ \foldl~(\mergeBy~(\odot))~[~]
  \circ \filterkey~k  =
    \foldl \odot z \circ \map~\val \circ \filterkey~k \mbox{,}
    \label{eq:head-fold-lift}
\end{equation}
where $\head_z$ returns $z$ when the input is empty. Finally,
in the equation below, given a RDD and any binary operator $(\odot)$,
the LHS computes $\foldl~(\mergeBy~(\odot))~[~])$ on each
partition, pick those with key $k$, and concatenates their values.
The RHS filters the values with key $k$, and computes $\foldl~(\odot)~z$
for each partition.
\begin{align}
  \MoveEqLeft
  \concat \circ \map~(\map~\val \circ \filterkey~k \circ
    \foldl~(\mergeBy~(\odot))~[~]) \nonumber \\
    {}={}&\map~(\foldl~(\odot)~z) \circ \filter~(\xnot \circ \xnull) \circ
       \map~(\map~\val \circ \filterkey~k) \mbox{.}
       \label{eq:mapfold-select}
\end{align}
All the lemmas above can be proved by induction. The proof of this lemma
follows:
  \begin{align*}
    \MoveEqLeft
    \lookUp~k \circ \aggBKey~\zero~(\xseq)~(\xcomb)
    \\
    {}={}&
    \head_z \circ \concat \circ \map~(\map~\val \circ \filterkey~k) \circ
     \repartition \circ {}
    \\
    &
    \foldl~(\mergeBy~(\xcomb))~[~] \circ \concat \circ
    \map~(\foldl~(\mergeBy~(\xseq))~[~]) \circ \perm
    &\expl{def.\ of $\aggBKey$}
    \\
    {}={}&
    \head_z \circ \map~\val \circ \filterkey~k \circ \foldl~(\mergeBy~(\xcomb))~[~] \circ {}
    \\
    &
    \concat \circ
    \map~(\foldl~(\mergeBy~(\xseq))~[~]) \circ \perm
    &\expl{naturality}
    \\
    {}={}&
    \head_z \circ \map~\val \circ \foldl~(\mergeBy~(\xcomb))~[~] \circ \filterkey~k \circ {}
    \\
    &
    \concat \circ \map~(\foldl~(\mergeBy~(\xseq))~[~]) \circ
    \perm
    &\expl{by \eqref{eq:filterkey-fold}}
    \\
    {}={}&
    \foldl~(\xcomb)~\zero \circ \map~\val \circ \filterkey~k \circ
    \concat \circ \map~(\foldl~(\mergeBy~(\xseq))~[~]) \circ \perm
    &\expl{by \eqref{eq:head-fold-lift}}
    \\
    {}={}&
    \foldl~(\xcomb)~\zero \circ \concat \circ
    \map~(\map~\val \circ \filterkey~k \circ \foldl~(\mergeBy~(\xseq)~[~]))
    \circ \perm
    &\expl{naturality}
    \\
    {}={}&
    \foldl~(\xcomb)~\zero \circ \map~(\foldl~(\xseq)~\zero) \circ
    \filter~(\xnot \circ \xnull) \circ \map~(\map~\val \circ \filterkey~k) \circ \perm
    &\expl{by \eqref{eq:mapfold-select}}
    \\
    {}={}&
    \foldl~(\xcomb)~\zero \circ \map~(\foldl~(\xseq)~\zero) \circ \perm \circ
    \filter~(\xnot \circ \xnull) \circ \map~(\map~\val \circ \filterkey~k)
    &\expl{naturality}
    \\
    {}={}&
    \aggWKey~k~\zero~(\xseq)~(\xcomb)
    &\expl{def.\ of $\aggWKey$}
&&~\qed
  \end{align*}
} 
\end{proof}

\lemAggByKeyIsAgg*


\begin{proof}
From Lemma~\ref{lem:lemAggBKeyIsAggWKey}, it follows that
$\aggBKey(\zero, \seq, \comboper, \prdd)$ has deterministic outcome iff
for all keys $k \in \alpha$ and partitionings $\partit$:
\begin{equation}
  \aggWKey(k, \zero, \seq, \comboper, \partit(\datalist)) =
  \foldl(\zero, \seq, \filterkey(k, \datalist)) .
\end{equation}

\noindent
From the defition of \hscode{\kw{aggregateWithKey}}, we infer that this is
equivalent to
%
\begin{align*}
  &&
  \agg(\zero, \seq, \comboper, \partit(\filterkey(k, \datalist)))
  &=
  \foldl(\zero, \seq, \filterkey(k, \datalist))
  \\
  \iff&&
  \agg(\zero, \seq, \comboper, \partit(\datalist'))
  &=
  \foldl(\zero, \seq, \datalist') ,
  &\expl{subst.\ $\datalist' = \filterkey(k, \datalist)$}
\end{align*}
which is the condition for $\agg(\zero, \seq, \comboper,
\partit(\datalist'))$ to have a~deterministic outcome.
\qed
\end{proof}


\noindent
Consider the following function.

\begin{lstlisting}[language=Haskell,basicstyle=\sffamily\footnotesize,morekeywords={reduce,reduceWithKey}]
reduceWithKey :: $\alpha$ $\rightarrow$ ($\beta$ $\rightarrow$ $\beta$ $\rightarrow$ $\beta$) $\rightarrow$ PairRDD $\alpha$ $\beta$ $\rightarrow$ $\beta$
reduceWithKey k mergeValue pairRdd =
  let select p = key p == k
     vrdd = filter (not . null)
              (map ((map value) . (filter select)) pairRdd)
  in reduce mergeValue vrdd
\end{lstlisting}

\begin{restatable}{lemma}{lemRedBKeyIsRedWKey}\label{lem:lemRedBKeyIsRedWKey}
It holds that
\begin{align*}
  \MoveEqLeft
  \lookup(k, \redBKey(\comboper, \prdd))
  = \redWKey(k, \comboper, \prdd)) .
\end{align*}
\end{restatable}

\begin{proof}
Similar to that of Lemma~\ref{lem:lemAggBKeyIsAggWKey}.
\qed
\end{proof}


\lemRedBKeyIsRed*

\begin{proof}
Folows the same structure as the proof of Proposition~\ref{lem:lemAggByKeyIsAgg}.
\qed
\end{proof}


\lemDetAggMsg*

\begin{proof}
When \hscode{reduceByKey} has deterministic outcome, then it holds (from
definition) that for all vertices $v \in \vertexid$, lists $\datalist \in
\listof{\alpha}$, and partitionings $\partit$:
\begin{align*}
  \lookup(v, \redListWKey(\partit, \comboper, \datalist))
  {}={}& \redl(\comboper, \filterkey(v, \datalist)) .
\end{align*}
When applying $\lookup(v, \aggMsgs(\send,\comboper,\graphrddof{\vertexlist,
\edgelist}))$, the result will be the same as if the $\lookup$ is applied to
the last line of function \hscode{\kw{aggregateMessagesWithActiveSet}}:
\begin{lstlisting}[language=Haskell,basicstyle=\sffamily\footnotesize,deletekeywords={},mathescape,morekeywords={reduceByKey}]
  lookup(v, reduceByKey($\comboper$, pairRdd)) .
\end{lstlisting}
Since $\redBKey(\comboper, \hscode{pairRdd})$ has deterministic outcome, it follows that
\begin{equation}
  \lookup(v, \redBKey(\comboper, \hscode{pairRdd})) = \redl(\comboper, \filterkey(v, \hscode{pairRdd})) .
\end{equation}
This is a~sufficent condition to conclude that
$\aggMsgs(\send,\comboper,\graphrddof{\vertexlist, \edgelist}))$ has
a~deterministic outcome.
\qed
\end{proof}



\end{document}